\documentclass[11pt]{article}
\usepackage{jheppub}
\usepackage{amsmath}
\usepackage[most]{tcolorbox}
\usepackage{dsfont}
\usepackage{ulem}
\usepackage{natbib}
\usepackage{xcolor}
\usepackage[hang,flushmargin]{footmisc}
\usepackage{tikz-cd}
\usepackage{enumitem}
\usepackage{amsfonts,amssymb, amscd,amsmath,latexsym,amsbsy,bm, amsthm}
\usepackage{stmaryrd}
\usepackage{todonotes}
\usepackage{float}
\newtheorem{theorem}{Theorem}[section]
\newtheorem{lemma}[theorem]{Lemma}
\newtheorem{definition}{Definition}[section]

\makeatletter\renewcommand{\@biblabel}[1]{#1.}\makeatother

\newtcolorbox{empheqboxed}{colback=gray!20, 
 colframe=white,
 width=\textwidth,
 sharpish corners,
 top=0mm, 
 bottom=0pt
}

\title{On Bailey pairs for \texorpdfstring{$\mathcal{N}=2$}{N=2} supersymmetric gauge theories on \texorpdfstring{$S_b^3/\mathbb{Z}_r$}{Sb3/Zr}}

\author{Ilmar Gahramanov$^{a,b,c}$, Batuhan Keskin$^{a, d}$, Dilara Kosva$^{a}$ and  Mustafa Mullahasanoglu$^{a}$ }

\affiliation{
	
	$^a$ {Department of Physics, Bogazici University, 34342 Bebek, Istanbul, Turkey}\\[-0.5cm]
	
	$^{b}$ Steklov Mathematical Institute of Russ. Acad. Sci.,
Gubkina str. 8, 119991 Moscow, Russia\\[-0.5cm]
	
	$^{c}$ Department of Mathematics, Khazar University,  Mehseti St. 41, AZ1096, Baku, Azerbaijan \\[-0.5cm]

 $^{d}$ Department of Electrical and Electronics Engineering,  Bogazici University, 34342 Bebek, Istanbul, Turkey\\[-0.5cm]

}

\emailAdd{ilmar.gahramanov@boun.edu.tr}
\emailAdd{batuhan.keskin@boun.edu.tr}
\emailAdd{dilara.kosva@boun.edu.tr}
\emailAdd{mustafa.mullahasanoglu@boun.edu.tr}

\abstract{We study Bailey pairs construction for hyperbolic hypergeometric integral identities acquired via the duality of lens partitions functions for the three-dimensional $\mathcal N=2$ supersymmetric gauge theories on $S_b^3/\mathbb{Z}_r$. The novel Bailey pairs are constructed for the star-triangle relation, the star-star relation, and the pentagon identity. The first two of them are integrability conditions for the Ising-type integrable lattice models. The last one corresponds to the representation of the basic $2-3$ Pachner move for triangulated 3-manifolds.
}


\begin{document}
\maketitle
\flushbottom

\section{Introduction}

In recent years, the remarkable concept of hypergeometric identities sits at the intersection of diverse studies such as exact results in supersymmetric gauge theories \cite{Gang:2019juz,Benini:2011nc,Imamura:2012rq,Imamura:2013qxa,Nieri:2015yia,Alday:2012au,Yamazaki:2013fva,Honda:2016vmv,Nedelin:2016gwu,Spiridonov:2009za,Spiridonov2014,Krattenthaler:2011da, Gahramanov:gka,Dolan:2011rp} and their mathematical structures interacting with various fields in mathematics, see, e.g., \cite{Gahramanov:2015tta, Tachikawa:2017byo, Spiridonov:2019kto, Gahramanov:2022qge}, star-triangle relation (Yang-Baxter equation) \cite{Spiridonov:2010em,Yamazaki:2012cp,Kels:2015bda,Yagi:2015lha,Gahramanov:2015cva,Gahramanov:2016ilb,Sarkissian:2018ppc,Eren:2019ibl,de-la-Cruz-Moreno:2020xop,Bozkurt:2020gyy} or star-star relation \cite{Yamazaki:2015voa,Kels:2017toi,Catak:2021coz, Mullahasanoglu:2021xyf} for spin lattice models, knot theory \cite{Kashaev:2012cz}, pentagon identities \cite{Kashaev:2012cz,Kashaev:2014rea,kashaev2014euler,Gahramanov:2013rda,Gahramanov:2014ona,Gahramanov:2016wxi,Bozkurt:2018xno,Jafarzade:2018yei, Dede:2022ofo}, Bailey pairs \cite{spiridonov2004bailey,Gahramanov:2015cva,Brunner:2017lhb,Spiridonov_2019,Gahramanov:2021pgu}, quantum algebras \cite{Hadasz:2013bwa,Bozkurt:2020gyy,Fan:2021bwt,Apresyan:2022erh}, etc.

In this work, we consider certain three-dimensional $\mathcal N=2$ supersymmetric dualities on $S^3_b/\mathbb{Z}_r$. These dualities have been studied via the gauge/YBE correspondence, which connects dualities with integrable models in statistical models, see for a comprehensive review \cite{Gahramanov:2017ysd,Yamazaki:2018xbx}. 

Here we consider three-dimensional $SU(2)$ gauge theory with six flavors $S^3_b/\mathbb{Z}_r$. The corresponding star-triangle relation for this theory was constructed in \cite{Gahramanov:2016ilb}. The reduction of the gauge symmetry to $U(1)$ via the gauge symmetry breaking  gives another solution to the star-triangle relation which is realized as the generalized Faddeev-Volkov model \cite{Bozkurt:2020gyy,Sarkissian:2018ppc}, since the $r=1$ case gives the Faddeev-Volkov model \cite{Bazhanov:2007mh,Bazhanov:2007vg} corresponding to $\mathcal N=2$ supersymmetric dual theories on $S^3_b$ \cite{Hama:2011ea}. The star-star relations \cite{baxter:1997ssr} of these two models are constructed, and for the factorized interaction-round-a-face (IRF) spin models,  IRF-type YBE are presented in \cite{Catak:2021coz,Mullahasanoglu:2021xyf}. The integral identity obtained by supersymmetric duality with  the gauge group $U(1)$ can also be written as a pentagon identity \cite{Bozkurt:2020gyy}.

We present new Bailey pairs for the hyperbolic hypergeometric integral identities. The construction of Bailey pairs for the star-triangle relation leads to acquiring the vertex-type\footnote{In the vertex-type models, the spins (continuous and discrete) are located at the edges, and the interactions of the spins are via vertices.} YBE via the Coxeter relations \cite{Gahramanov:2015cva}. Bailey pairs for the star-star relations and the pentagon identity mentioned above are also constructed.

The organization of the rest of this paper is as follows. In Section 2, we briefly recall the mathematical tools and introduce integral identities resulting from the equality of the partition functions of $\mathcal{N}=2$ supersymmetric gauge theories on $S^3_b/\mathbb{Z}_r$. In Section 3, we present the star-triangle relations, the star-star relations, and the pentagon identities obtained via the supersymmetric dualities. In Section 4, we construct Bailey pairs for each integral identity considered in previous sections. In section 5, we conclude our results and present some further studies. The appendix consists of the construction of the Bailey pair for the star-triangle relation. 


\section{Seiberg dualities on \texorpdfstring{$S_b^3/\mathbb{Z}_r$}{Sb3/Zr}}

In this section, we consider Seiberg dualities \cite{Seiberg:1994pq} in three-dimensional $\mathcal N=2$ theories \cite{Intriligator:1996ex, Aharony:1997bx}. One of the evidences\footnote{The equality of the  superconformal indices is also evidence for duality, see e.g. \cite{Krattenthaler:2011da,Kapustin:2011jm,Gahramanov:2013rda,Gahramanov:2016wxi}} for dualities is the equality of partition functions\footnote{Supersymmetric partition functions are studied on sphere (e.g. \cite{Kapustin:2010xq}), squashed sphere (e.g. \cite{Dolan:2011rp,Gahramanov:gka,Amariti:2015vwa}), and squashed lens space  (e.g. \cite{Benini:2011nc, Imamura:2012rq,Imamura:2013qxa})}.

The three-dimensional $\mathcal N=2$ partition functions on the squashed lens space $S^3_b/\mathbb{Z}_r$ have been computed via dimensional reduction of the four-dimensional lens superconformal index \cite{Benini:2011nc,Yamazaki:2013fva,Eren:2019ibl} and via the supersymmetric localization technique \cite{Imamura:2012rq,Imamura:2013qxa}. Such theories have been studied in \cite{Nieri:2015yia,Gahramanov:2016ilb, Bozkurt:2020gyy,Catak:2021coz, Mullahasanoglu:2021xyf}.	

\subsection{3d \texorpdfstring{$SU(2)$}{SU(2)} duality}

Our starting point is a three-dimensional $\mathcal N=2$ $SU(2)$ gauge theory with six fundamentals and six anti-fundamentals. The confined dual theory consists of only fifteen chiral multiplets in the totally antisymmetric tensor representation of the flavor group. The equality of the partition functions can be written as an integral identity in terms of hyperbolic hypergeometric functions\footnote{For the integral identity written in terms of the improved double sine function, see, e.g., \cite{Gahramanov:2016ilb}}, 
\begin{align}
\begin{aligned}
\label{SU2identity}
	 \sum_{y=0}^{[ r/2 ]}\epsilon (y) \int _{-\infty}^{\infty} 
& \frac{\prod_{i=1}^6\gamma^{(2)}(-i(a_i\pm z)-i\omega_1(u_i\pm y);-i\omega_1r,-i\omega)}{\gamma^{(2)}(\pm 2iz\pm 2i\omega_1y;-i\omega_1r,-i\omega)}  \\
	\times & \frac{\gamma^{(2)}(-i(a_i\pm z)-i\omega_2(r-(u_i\pm y));-i\omega_2r,-i\omega)}{\gamma^{(2)}(\pm2iz-i\omega_2(r\pm2y);-i\omega_2r,-i\omega)} 
	\frac{dz}{2r\sqrt{-\omega_1\omega_2}}
    \\ =\prod_{1\leq i<j\leq 6} & \gamma^{(2)}(-i(a_i + a_j)-i\omega_1(u_i + u_j);-i\omega_1r,-i\omega)  \\ & \times 
	\gamma^{(2)}(-i(a_i + a_j)-i\omega_2(r-(u_i + u_j));-i\omega_2r,-i\omega)\; ,
\end{aligned}
\end{align}	
with the balancing conditions\footnote{When the balancing condition is taken $\sum_{i=1}^6 u_i=mr$ where $m$ is an integer, there should be a sign factor $ e^{\frac{\pi i}{2}mr^2(m-1)(2m-1)}$ on the left-hand side of the integral identity (\ref{SU2identity}), seem \cite{Gahramanov:2016ilb} .} are $\sum^6_{i=1} a_i=\omega$ and\space$\sum_{i=1}^6 u_i=0$, where we introduced  $\omega :=\omega_1+\omega_2$. Also, $\epsilon (0)=\epsilon ({[ r/2 ]})=1$ and $\epsilon(y)=2$ otherwise.  Here the function $\gamma^{(2)}(z; \omega_1, \omega_2)$ is the so-called hyperbolic gamma function\footnote{Different versions of the hyperbolic gamma function can be seen as the double sine function \cite{Kharchev:2001rs,Ponsot:2000mt}, the non-compact quantum dilogarithm \cite{Faddeev:1999fe,Faddeev:2000if,Ponsot:2000mt,Volkov:2005zrq,Bazhanov:2010kz}, the modified $q$-gamma
function \cite{Spiridonov-essays}, etc.} \cite{Ruijsenaars:1997:FOA,van2007hyperbolic, Andersen:2014aoa} which is the main tool in this study. 
One of the several representations\footnote{For various integral representations, see, e.g. \cite{Faddeev:1995nb,woronowicz2000quantum} and one can also introduce the infinite product representation
	\begin{align}
	\gamma^{(2)}(z;\omega_{1},\omega_{2})=e^{\frac{\pi i}{2}B_{2,2}(z;\omega_{1},\omega_{2})}\frac{(e^{-2\pi i\frac{z}{\omega_{2}}}\tilde{q};\tilde{q})}{(e^{-2\pi i\frac{z}{\omega_{1}}};q)} \; ,
	\end{align}
	where parameters are $\tilde{q}=e^{2\pi i \omega_{1}/\omega_{2}}$ and $q=e^{-2\pi i \omega_{2}/\omega_{1}}$ and
	 the Bernoulli polynomial is
	\begin{align}
	B_{2,2}(z;\omega_{1},\omega_{2})=\frac{z^2-z(\omega_{1}+\omega_{2})}{\omega_{1}\omega_{2}}+\frac{\omega_{1}^2+3\omega_{1}\omega_{2}+\omega_{2}^2}{6\omega_{1}\omega_{2}}.
	\end{align}
	}  
of this special function is the following\footnote{One can list many areas of study for this function, but we mention fewer examples from the areas of mathematical and theoretical physics such as knot theory \cite{hikami2001hyperbolic,Hikami2007,hikami2014braiding,Chan:2017qnw} supersymmetric gauge theory \cite{Teschner:2012em} integrable models of statistical mechanics \cite{Bazhanov:2007mh,Bazhanov:2007vg} special functions \cite{van2007hyperbolic}.} 
	\begin{align}
	\gamma^{(2)}(z;\omega_{1},\omega_{2})=\exp{\left(-\int_{0}^{\infty}\frac{dx}{x}\left[\frac{\sinh{x(2z-\omega_{1}-\omega_{2})}}{2\sinh{(x\omega_{1})}\sinh{(x\omega_{2})}}-\frac{2z-\omega_{1}-\omega_{2}}{2x\omega_{1}\omega_{2}}\right]\right)} \; ,
	\end{align}
	where $Re(\omega_{1}),Re(\omega_{2})>0$ and $Re(\omega_{1}+\omega_{2})>Re(z)>0$. 
We will mainly use the reflection property of the hyperbolic gamma function
\begin{align}
    \gamma^{(2)}(z;\omega_{1},\omega_{2})\gamma^{(2)}(\omega_{1}+\omega_{2}-z;\omega_{1},\omega_{2})=1 \:.
\end{align}
and the following shorthand notation
\begin{align}
    \gamma^{(2)}(\pm z;\omega_{1},\omega_{2})=\gamma^{(2)}( z;\omega_{1},\omega_{2})\gamma^{(2)}(-z;\omega_{1},\omega_{2}) \:.
\end{align}

Note that the case $r=1$ (see, e.g., \cite{van2007hyperbolic}) of the integral identity (\ref{SU2identity}) corresponds to the duality of supersymmetric gauge theories on $S^3_b$.

\subsection{3d \texorpdfstring{$U(1)$}{U(1)} duality}
	
One obtains the following integral identity via breaking the gauge symmetry \cite{Bozkurt:2020gyy} (see also \cite{Spiridonov:2010em, Sarkissian:2018ppc}) from $SU(2)$ to $U(1)$ in the duality (\ref{SU2identity})
\begin{align}
\begin{aligned}
\label{U1identity}
    \sum_{y=0}^{[ r/2 ]}\epsilon (y) e^{\frac{\pi i}{2}C}\int _{-\infty}^{\infty} \prod_{i=1}^3 & \gamma^{(2)}(-i(a_i-z)-i\omega_1(u_i- y);-i\omega_1r,-i\omega)  \\
	\times & \gamma^{(2)}(-i(a_i-z)-i\omega_2(r-(u_i- y));-i\omega_2r,-i\omega)  \\
	\times & \gamma^{(2)}(-i(b_i+z)-i\omega_1(v_i+ y);-i\omega_1r,-i\omega) 
	\\
	\times & \gamma^{(2)}(-i(b_i+z)-i\omega_2(r-(v_i+y));-i\omega_2r,-i\omega)
	\frac{dz}{r\sqrt{-\omega_1\omega_2}}\\
	= \prod_{i,j=1}^3& \gamma^{(2)}(-i(a_i+ b_j)-i\omega_1(u_i+ v_j);-i\omega_1r,-i\omega)\\ \times& \gamma^{(2)}(-i(a_i+ b_j)-i\omega_2(r-(u_i+ v_j));-i\omega_2r,-i\omega))\;, 
	\end{aligned}
	\end{align}	
where the balancing conditions are $\sum_{i=1}^3 a_i+b_i=\omega$ and $\sum_{i=1}^3 u_i+v_i=0$ and the sign factor is $C=-2y+\sum_{i=1}^3(u_i-v_i)$. It is possible to shift the discrete parameters $u_i$ and $v_i$ and obtain a new  balancing condition $\sum_{i=1}^3 u_i+v_i=r$. In this case $C=0$  in (\ref{U1identity}). 




\section{Integrability conditions and the basic 2-3 Pachner move }

\subsection{Star-triangle relation }

In the transfer matrix method \cite{Baxter:1982zz} for Ising-type models, it is sufficient to write a star-triangle relation to obtain the integrability property of the lattice spin model. Here we are interested in Ising-like models with discrete $m_i$ and continuous $x_i$ spin variables. We denote the discrete and continuous spins together in the form of $\sigma_i:=(x_i,m_i)$.

The Boltzmann weights of the models discussed here have the reflection property $W(\sigma_i, \sigma_j)= W(\sigma_j, \sigma_i)$ and the crossing symmetry $\overline{W}_{\alpha_i,\Tilde{\alpha}_j}(\sigma_i,\sigma_j)=W_{\eta-\alpha_i,\beta-\Tilde{\alpha}_j}(\sigma_i,\sigma_j)$, which means that one can write vertical interactions in terms of horizontal interactions. Hence, we can write the star-triangle relation as the following 
\begin{align}
\begin{aligned}
   \sum_{m_0} \int dx_0 \: S(\sigma_0) \: W_{\alpha_1,\Tilde{\alpha}_1}(\sigma_1,\sigma_0)W_{\alpha_2,\Tilde{\alpha}_2}(\sigma_2,\sigma_0)W_{\alpha_3,\Tilde{\alpha}_3}(\sigma_3,\sigma_0) 
    \makebox[10em]{}
   \\ \makebox[6em]{}
   =\mathcal{R} \: W_{\eta-\alpha_1,\beta-\Tilde{\alpha}_1}(\sigma_1,\sigma_2)W_{\eta-\alpha_2,\beta-\Tilde{\alpha}_2}(\sigma_1,\sigma_3)W_{\eta-\alpha_3,\beta-\Tilde{\alpha}_3}(\sigma_2,\sigma_3) 
   \label{str}\:,
\end{aligned}
	\end{align}
where the constraints on the spectral parameters are $\alpha_1+\alpha_2+\alpha_3=\eta$ (note that these are continuous and in our cases $\eta=-\frac{\omega}{2}$) and $\Tilde{\alpha}_1+\Tilde{\alpha}_2+\Tilde{\alpha}_3=\beta$ (note that these are discrete and in our cases $\beta=0$) with crossing parameters $\eta$ and $\beta$. The functions $S(\sigma_0)$ and $\mathcal{R}$ are the self-interaction contribution and the spin-independent functions, respectively.

The identity (\ref{SU2identity}) turns to the star-triangle relation\footnote{In \cite{Gahramanov:2016ilb}, it is firstly appeared as normalized and with only continuous spectral parameters.} when new variables 
\begin{align}
    	\begin{aligned}
	a_i & =-\alpha_i+x_{i}\;, ~~~~~~~~~~~~~\; a_{i+3}=-\alpha_i-x_{i}  \;, \\
	u_i & =-\Tilde{\alpha}_i+y_{i}\;, ~~~~~~~~~~~~~\; u_{i+3}=-\Tilde{\alpha}_i-y_{i}\:,
	\end{aligned}
	\end{align}
are introduced. The Boltzmann weights can be written as  
\begin{equation}
\begin{aligned}\label{b1}
    W_{\alpha_i,\Tilde{\alpha}_i}(x_i,x_j&,y_i,y_j)=
    \\&
    \gamma^{(2)}(-i(-\alpha_i+x_i\pm x_j)-i\omega_1(-\Tilde{\alpha}_i+y_i\pm y_j);-i\omega_{1}r,-i\omega)  \\&\times 
    \gamma^{(2)}(-i(-\alpha_i+x_i\pm x_j)-i\omega_2(r-(-\Tilde{\alpha}_i+y_i\pm y_j));-i\omega_2r,-i\omega)  \\&\times 
    \gamma^{(2)}(-i(-\alpha_i-x_i\pm x_j)-i\omega_1(-\Tilde{\alpha}_i-y_i\pm y_j);-i\omega_{1}r,-i\omega)  \\&\times 
    \gamma^{(2)}(-i(-\alpha_i-x_i\pm x_j)-i\omega_2(r-(-\Tilde{\alpha}_i-y_i\pm y_j));-i\omega_2r,-i\omega) \;.
\end{aligned}
	\end{equation}
	
The spin-independent weight function $\mathcal{R}$ depending only on spectral parameters is the same for both models
\begin{align}
	\mathcal{R}=\prod_{j=1}^3\gamma^{(2)}(2i\alpha_j+2i\omega_1\Tilde{\alpha}_j;-i\omega_{1}r,-i\omega)\gamma^{(2)}(2i\alpha_j-i\omega_{2}(r+2\Tilde{\alpha}_j);-i\omega_{2}r,-i\omega)\; . \label{R}
\end{align}
	
However, a self-interaction contribution is not trivial in the $SU(2)$ model (the identity (\ref{SU2identity})) and has the following form
\begin{equation}
\begin{aligned}
 S(\sigma_0) = \frac{ \epsilon (n)}{\gamma^{(2)}(\pm 2iu \pm   2i\omega_1n;-i\omega_1r,-i\omega)\gamma^{(2)}(\pm2iu-i\omega_2(r\pm2n);-i\omega_2r,-i\omega)}\:,
\end{aligned}
\end{equation}
where $\epsilon (n)$ disappears if one changes boundaries of the summation as in \cite{Bozkurt:2020gyy} (see Appendix in \cite{Eren:2019ibl}). 


The same procedure is applied to (\ref{U1identity})\footnote{This star-triangle relation for generalized Faddeev-Volkov model appeared in \cite{Sarkissian:2018ppc,Bozkurt:2020gyy} has only continuous spectral parameters.  } by re-defining variables as 
\begin{align}
    	\begin{aligned}
	a_i & =-\alpha_i+x_{i}\;, ~~~~~~~~~~~~~\; b_{i}=-\alpha_i-x_{i}  \;, \\
	u_i & =-\Tilde{\alpha}_i+y_{i}\;, ~~~~~~~~~~~~~\; v_{i}=-\Tilde{\alpha}_i-y_{i}\:,
	\end{aligned}
	\end{align}
then, the Boltzmann weights become
\begin{align}
\begin{aligned} \label{b2}
    W_{\alpha_i,\Tilde{\alpha}_i}(x_i,x_j&,y_i,y_j) = \\&
    e^{-\pi i (y_i+y_j)}
    \gamma^{(2)}(-i(-\alpha_i+x_i- x_j)-i\omega_1(-\Tilde{\alpha}_i+y_i- y_j);-i\omega_{1}r,-i\omega)  \\&\times 
    \gamma^{(2)}(-i(-\alpha_i+x_i- x_j)-i\omega_2(r-(-\Tilde{\alpha}_i+y_i- y_j));-i\omega_2r,-i\omega)  \\&\times 
    \gamma^{(2)}(-i(-\alpha_i-x_i+ x_j)-i\omega_1(-\Tilde{\alpha}_i-y_i+ y_j);-i\omega_{1}r,-i\omega)  \\&\times 
    \gamma^{(2)}(-i(-\alpha_i-x_i+ x_j)-i\omega_2(r-(-\Tilde{\alpha}_i-y_i+ y_j));-i\omega_2r,-i\omega) \;,
\end{aligned}
	\end{align}
where the exponent term vanishes if we change the balancing condition mentioned in (\ref{U1identity}). This model has no self-interaction term and the spin-independent function is the same as (\ref{R}).

\subsection{Star-star relation }

One can obtain another fundamental integrability condition in statistical mechanics which is the star-star relation\footnote{For the star-star relation in the context of supersymmetric partition functions for dualities, and symmetries of beta hypergeometric integrals, see  \cite{Spiridonov:2008zr, Dimofte:2012pd, Kels:2017toi}.} \cite{baxter:1997ssr} in the existence of star-triangle relation. In some lattice spin models, Boltzmann weights satisfy the star-star relation but not the star-triangle relation, see e.g. \cite{baxter:1997ssr, Bazhanov:2011mz, Bazhanov:2013bh}. 

The star-star relation has the following form
\begin{align}
\begin{aligned}
&  \sum_{m_{0}} \int d x_{0} W_{\alpha_1,\Tilde{\alpha}_1}\left(\sigma_{1}, \sigma_{0}\right) W_{\alpha_2,\Tilde{\alpha}_2}\left(\sigma_{0}, \sigma_{2}\right) W_{\alpha_3,\Tilde{\alpha}_3}\left(\sigma_{3}, \sigma_{0}\right) W_{\alpha_4,\Tilde{\alpha}_4}\left(\sigma_{0}, \sigma_{4}\right) 
  \\
&\quad \quad \quad\quad=\frac{W_{2\eta-\alpha_1-\alpha_2,2\beta-\Tilde{\alpha}_1-\Tilde{\alpha}_2}(\sigma_1,\sigma_2)W_{2\eta-\alpha_1-\alpha_4,2\beta-\Tilde{\alpha}_1-\Tilde{\alpha}_4}(\sigma_1,\sigma_4)}
{W_{2\eta-\alpha_3-\alpha_4,\beta-\Tilde{\alpha}_3-\Tilde{\alpha}_4}(\sigma_4,\sigma_3)W_{2\eta-\alpha_2-\alpha_3,2\beta-\Tilde{\alpha}_2-\Tilde{\alpha}_3}(\sigma_2,\sigma_3)}
  \\ & \quad\quad\quad\quad\times 
\sum_{m_{0}} \int d x_{0} W_{\alpha_3,\Tilde{\alpha}_3}\left(\sigma_{0},\sigma_{1}\right) W_{\alpha_4,\Tilde{\alpha}_4}\left(\sigma_{2},\sigma_{0}\right) W_{\alpha_1,\Tilde{\alpha}_1}\left(\sigma_{0},\sigma_{3}\right) W_{\alpha_2,\Tilde{\alpha}_2}\left(\sigma_{4},\sigma_{0}\right) \:,
    \label{gf}
\end{aligned}
	\end{align}
where the spectral parameters satisfy the conditions $\sum_{i=1}^4\alpha_i=2\eta$ and $\sum_{i=1}^4\Tilde{\alpha}_i=2\beta$.

Using the hyperbolic hypergeometric integral identity (\ref{SU2identity}) (the star-triangle relation), one can obtain the following integral identity \cite{Mullahasanoglu:2021xyf} presented as a star-star relation
\begin{align}
\begin{aligned}
      \sum_{y=0}^{[ r/2 ]}\epsilon (y) \int _{-\infty}^{\infty} \frac{\prod_{i=1}^8 \gamma^{(2)}(-i(a_i\pm z)-i\omega_1(u_i\pm y);-i\omega_1r,-i\omega)}{\gamma^{(2)}(\pm 2iz\pm i\omega_12y;-i\omega_1r,-i\omega)}   \\
     \times \frac{\gamma^{(2)}(-i(a_i\pm z)-i\omega_2(r-(u_i\pm y));-i\omega_2r,-i\omega)}{\gamma^{(2)}(\pm2iz-i\omega_2(r\pm2y);-i\omega_2r,-i\omega)} 
     \frac{dz}{r\sqrt{-\omega_1\omega_2}}
       \\ 
     =
     \frac{\prod_{1\leq i<j\leq 4}\gamma^{(2)}(-i(a_i + a_j)-i\omega_1(u_i + u_j);-i\omega_1r,-i\omega) }
     {\prod_{5\leq i<j\leq 8}\gamma^{(2)}(-i(\Tilde{a}_i+\Tilde{a}_j)-i\omega_1(\Tilde{u}_i+\Tilde{u}_j);-i\omega_1r,-i\omega)  }   \\
    \times \frac{\prod_{1\leq i<j\leq 4}  \gamma^{(2)}(-i(a_i + a_j)-i\omega_2(r-(u_i + u_j));-i\omega_2r,-i\omega)}
     {\prod_{5\leq i<j\leq 8} \gamma^{(2)}(-i(\Tilde{a}_i+\Tilde{a}_j)-i\omega_2(r-(\Tilde{u}_i+\Tilde{u}_j));-i\omega_2r,-i\omega)}   \\
     \times   \sum_{m=0}^{[ r/2 ]}\epsilon (m) \int _{-\infty}^{\infty} \frac{\prod_{i=1}^8\gamma^{(2)}(-i(\Tilde{a}_i\pm x)-i\omega_1(\Tilde{u}_i\pm m);-i\omega_1r,-i\omega)}{\gamma^{(2)}(\pm 2ix\pm i\omega_12m;-i\omega_1r,-i\omega)}   \\
     \times \frac{\gamma^{(2)}(-i(\Tilde{a}_i\pm x)-i\omega_2(r-(\Tilde{u}_i\pm m));-i\omega_2r,-i\omega)}{\gamma^{(2)}(\pm2ix-i\omega_2(r\pm2m);-i\omega_2r,-i\omega)}
     \frac{dx}{r\sqrt{-\omega_1\omega_2}}
      \:,\label{newssr}
\end{aligned}
	\end{align}
with the balancing conditions $\sum_{i=1}^8a_i=2\omega$ and $\sum_{i=1}^8u_i=0$, and parameters are identified as 
\begin{align}
\begin{aligned}
     \tilde{a}_i & =  a_i+s, & \tilde{u}_i & =u_i+p,  & \text{if} \;\;\; i=1,2,3,4 \:,
     \\ 
     \tilde{a}_i & =  a_i-s,  & \tilde{u}_i & = u_i-p, & \text{if} \;\;\; i=5,6,7,8\:,
\end{aligned}
\end{align}
where 
\begin{align}
\begin{aligned}
s & =\frac{1}{2}\left(\omega-\sum_{i=1}^4a_i\right)=\frac{1}{2}\left(-\omega_1-\omega_2+\sum_{i=5}^8a_i\right)\:, \\
p & =-\frac{1}{2}\left(\sum_{i=1}^4u_i\right)=\frac{1}{2}\left(\sum_{i=5}^8u_i\right) \:.
\end{aligned}
\end{align}

The following identity is the star-star relation \cite{Catak:2021coz} of the generalized Faddeev-Volkov model 
\begin{align}
\begin{aligned}
  \sum_{y=0}^{[ r/2 ]}\epsilon (y) \int _{-\infty}^{\infty} \prod_{i=1}^4  \gamma^{(2)}(-i(a_i-z)-i\omega_1(u_i- y);-i\omega_1r,-i\omega)   \\
     \times  \gamma^{(2)}(-i(a_i-z)-i\omega_2(r-(u_i- y));-i\omega_2r,-i\omega)   \\
    \times  \gamma^{(2)}(-i(b_i+z)-i\omega_1(v_i+ y);-i\omega_1r,-i\omega)  
    \\
     \times  \gamma^{(2)}(-i(b_i+z)-i\omega_2(r-(v_i+y));-i\omega_2r,-i\omega)  
     \frac{dz}{r\sqrt{-\omega_1\omega_2}}  \\
     =\frac{e^{\frac{\pi i}{2}\sum_{i=1}^2(u_i-v_i)}}{e^{\frac{\pi i}{2}\sum_{i=3}^4(\tilde{u_i}-\tilde{v_i})}}
   \frac{\prod_{i,j=1}^2 \gamma^{(2)}(-i(a_i+ b_j)-i\omega_1(u_i+ v_j);-i\omega_1r,-i\omega) }{\prod_{i,j=3}^4 \gamma^{(2)}(-i(\tilde{a_i}+\tilde{b_i})-i\omega_1(\tilde{u_i}+\tilde{v_i});-i\omega_1r,-i\omega) } 
   \\ \times  
     \frac{\prod_{i,j=1}^2  \gamma^{(2)}(-i(a_i+ b_j)-i\omega_2(r-(u_i+ v_j));-i\omega_2r,-i\omega)}{\prod_{i,j=3}^4  \gamma^{(2)}(-i(\tilde{a_i}+\tilde{b_i})-i\omega_2(r-(\tilde{u_i}+\tilde{v_i}));-i\omega_2r,-i\omega) } 
    \\ \times  
     \sum_{m=0}^{[ r/2 ]}\epsilon (m) \int _{-\infty}^{\infty} \prod_{i=1}^4  \gamma^{(2)}(-i(\tilde{a_i}-x)-i\omega_1(\tilde{u_i}- m);-i\omega_1r,-i\omega)   \\
     \times  \gamma^{(2)}(-i(\tilde{a_i}-x)-i\omega_2(r-(\tilde{u_i}- m));-i\omega_2r,-i\omega)   \\
    \times  \gamma^{(2)}(-i(\tilde{b_i}+x)-i\omega_1(\tilde{v_i}+ m);-i\omega_1r,-i\omega) 
    \\
     \times  \gamma^{(2)}(-i(\tilde{b_i}+x)-i\omega_2(r-(\tilde{v_i}+m));-i\omega_2r,-i\omega) \frac{dx}{r\sqrt{-\omega_1\omega_2}}  \:,
\end{aligned}
	\end{align}
where the balancing conditions are $\sum_{i=1}^4a_i+b_i=2\omega$ and $\sum_{i=1}^4u_i+v_i=0$,
and we used the following choice of parameters,
\begin{align}
\begin{aligned}
     \tilde{a}_i & =  a_i+s, & \tilde{b}_i & =  b_i+s, & \tilde{u}_i & =u_i+p, & \tilde{v}_i  & = v_i+p, & \text{if} \;\;\; i=1,2\:,
     \\ 
     \tilde{a}_i & =  a_i-s, & \tilde{b}_i & =  b_i-s, & \tilde{u}_i & = u_i-p, & \tilde{v}_i & =v_i-p, & \text{if} \;\;\; i=3,4\:,
\end{aligned}
\end{align}
where 
\begin{align}
\begin{aligned}
s & =\frac{1}{2}(\omega+a_1+a_2+b_1+b_2)=\frac{1}{2}(\omega-a_3-a_4-b_3-b_4)\:, \\
p & =-\frac{1}{2}(u_1+u_2+v_1+v_2)=\frac{1}{2}(u_3+u_4+v_3+v_4) \:.
\end{aligned}
\end{align}

\subsection{The pentagon identity }
Pentagon relation \cite{Kashaev:2014rea,kashaev2014euler} has a meaning of the 2-3 Pachner move \cite{pachner1991pl} for triangulated three-dimensional manifolds and can be formally written as
\begin{align}
	 \mathcal{B}\mathcal{B}\mathcal{B}=\mathcal{B}\mathcal{B} \:.
	\label{pentagondefn}
	\end{align}

Here we write the equation (\ref{U1identity}) as an integral pentagon identity \cite{Bozkurt:2020gyy}. It can be interpreted as a topological invariant of corresponding 3–manifold via \textit{3d–3d} correspondence \cite{Dimofte:2011ju, Dimofte:2011py} building bridges between three dimensional $\mathcal{N}=2$ supersymmetric gauge theories and triangulated 3-manifolds.

Then the hyperbolic hypergeometric solution to the pentagon identity can be obtained by the following definition
\begin{align}
\begin{aligned}
	\mathcal{B}(z_1,u_1;z_2,u_2)=\frac{\gamma^{(2)}(-iz_1-i\omega_1u_1;-i\omega_1r,-i\omega)\gamma^{(2)}(-iz_1-i\omega_2(r-u_1);-i\omega_2r,-i\omega))}{\gamma^{(2)}(-i(z_1+z_2)-i\omega_1(u_1+u_2);-i\omega_1r,-i\omega))}   \\\frac{\gamma^{(2)}(-iz_2-i\omega_1u_2;-i\omega_1r,-i\omega)\gamma^{(2)}(-iz_2-i\omega_2(r-u_2);-i\omega_2r,-i\omega)}{\gamma^{(2)}(-i(z_1+z_2)-i\omega_2(r-u_1-u_2);-i\omega_2r,-i\omega)} \; ,
	\label{pentagon_b}
\end{aligned}
	\end{align}
and the equation (\ref{U1identity}) turns to the integral pentagon identity 
	\begin{align}
\begin{aligned}
	\frac{1}{r\sqrt{-\omega_1\omega_2}}\sum_{y=0}^{\lfloor r/2 \rfloor} e^{\frac{\pi iC}{2}} \int _{-\infty}^{\infty} dz \prod_{i=1}^3\mathcal{B}(a_i-z,u_i-y;b_i+z,v_i+y) \\=\mathcal{B}(a_1+b_2,u_1+v_2;a_2+b_3,u_2+v_3)\mathcal{B}(a_1+b_3,u_1+v_3;a_2+b_1,u_2+v_1) \:,
	\label{pentagon}
\end{aligned}
	\end{align}
where the sign factor and the balancing conditions are the same as in (\ref{U1identity}).

\section{Bailey pairs}
Influenced by Rogers' work in proving combinatorial identities which are now known as the Rogers-Ramanujan identities, W.N. Bailey introduced the following lemma to abstract the notions underlying the proofs \cite{bailey1948identities},

\begin{lemma}\label{bailey transform}
If the series $\{\alpha\}_{n\geq 0}$, $\{\beta\}_{n\geq 0}$, $\{\delta\}_{n\geq 0}$, $\{\gamma\}_{n\geq 0}$, $\{u\}_{n\geq 0}$ and $\{v\}_{n\geq 0}$ satisfy
\begin{equation*}
\beta_{n}=\sum_{r=0}^{n} \alpha_{r} u_{n-r} v_{n+r} \;,
\end{equation*}
and
\begin{equation*}
\gamma_{n}=\sum_{r=n}^{\infty} \delta_{r} u_{r-n} v_{r+n} \;,
\end{equation*}
then
\begin{equation*}
\sum_{n=0}^{\infty} \alpha_{n} \gamma_{n}=\sum_{n=0}^{\infty} \beta_{n} \delta_{n} \;. 
\end{equation*}
\end{lemma}
The proof relies on a simple rearrangement of the series and is trivial as noted by Bailey, hence will be omitted. Theorem $\eqref{bailey transform}$ is commonly referred to as the Bailey lemma and for specific choices of the mentioned series, various identities in mathematics can be derived \cite{slater1966generalized,bressoud1981, dbsearsbaileytransform, Schilling1997}.

Following Bailey's work, Andrews \cite{geandrews} formulated a method of deriving infinitely many identities from a known one iteratively. Given two sequences of functions $\{\alpha\}_{k}$ and $\{\beta\}_{k}$ for $k \in \{0,\dots,n \}$, with the relation
\begin{equation}\label{bailey def}
    \beta_{k} = F_k(\alpha_0, \dots, \alpha_n)\;,
\end{equation}
one can construct the functions 
\begin{align}
\beta_k^{(i)} &= G_k\left(\beta_0^{(i-1)}, \dots, \beta_n^{(i-1)}\right)\:,  \hspace{1cm} i\in \mathbb{Z}_{>0 }\:,\\
    \alpha_k^{(i)} &= H_k\left(\alpha_0^{(i-1)}, \dots, \alpha_n^{(i-1)}\right)\:, \hspace{1cm} i\in \mathbb{Z}_{>0 }\:,
\end{align}
that also satisfy $\eqref{bailey def}$ for every $i\in \mathbb{Z}_{>0 }$. Then the functions $\{\alpha^{(i)}\}_{k}$, $\{\beta^{(i)}\}_{k}$ are called a Bailey pair, and they form a chain of infinite length. This notion  can be generalized from chains to lattices and higher-dimensional chains, see, e.g. \cite{milne1992al,lilly1993thec,milne1995consequences}. Apart from their purely mathematical implications, Bailey pairs are also used in superconformal field theories \cite{Berkovich1996,Andrews1998scft} and exactly solvable models of statistical mechanics.
The latter will be the focus of our interest in the next sections.  For a detailed study of the history of the Bailey lemma, see \cite{warnaarhistory}.  

\subsection{Star-triangle relation}

\begin{definition}
The functions $\alpha(x,m;t,p)$ and $\beta(x,m; t,p)$, $x\in \mathbb{C}$, $m \in \mathbb{Z}$ form an integral Bailey pair with respect to parameters $t\in \mathbb{C}$ and $p \in \mathbb{Z}$ if the following relation is satisfied,
\begin{equation}\label{definition}
    \beta(z,m;t,p) = M(t,p)_{z,m;x,j}\alpha(x,j;t,p) \;, 
\end{equation}
where $M(t,p)_{z,m;x,j}$ is an integral-sum operator that integrates and sums over the continuous variable $x\in \mathbb{C}$ and the discrete variable $j\in \mathbb{Z}$ of $\alpha(x,j;t,p)$, respectively. 
\end{definition}

Suppose we have another operator $ D(s,q; y,l; x,k)$ of continuous $s,y,x \in \mathbb{C}$ and discrete variables $q,l,k\in \mathbb{Z}$, such that it satisfies the relation \begin{equation}
    D(s,q; y,l;x,k)  D(-s,-q;y,l;x,k)=1
\end{equation}
that we will refer to as the reflection relation, and $D(0,0; y,l;x,k)=1$. Moreover, we assume that the operators $M$ and $D$ satisfy the "star-triangle relation", given as
\begin{equation}\label{star triangle}
\begin{gathered}
       M(s,q)_{w,k; z,m}  D(s+t,q+p; y,l; z,m)  M(t,p)_{z,m; x,j} =\\
       D(t,p; y,l; w,k) M(s+t,q+p)_{w,k; x,j}  D(s,q; y,l; x,j)\:.
\end{gathered}
\end{equation}
Utilizing $M$ and $D$ operators, the next lemma addresses the question of forming infinitely many Bailey pairs after finding a particular one.

\begin{lemma}[Bailey Lemma]
Suppose $\alpha(x,m;t,p)$ and $\beta(x,m; t,p)$ form an integral Bailey pair with respect to $t \in \mathbb{C}$ and $p \in \mathbb{Z}$. Then, the sequences of functions $\alpha'(x,k;t+s,p+q)$ and $\beta'(x,k; t+s,p+q)$, $k \in \mathbb{Z}$, defined by 
\begin{align}
    \alpha'(x,k;t+s,p+q) &= D(s,q; y,l; x,k)\alpha(x,k; t,p) \:,\\
    \beta'(x,k; t+s,p+q)&= D(-t,-p;y,l;x,k)M(s,q)_{x,k;z,m}D(s+t,p+q;y,l;z,m) \beta(z,m; t,p)\:,
\end{align}
form a Bailey pair with respect to the new parameters $t+s$ and $p+q$ where $s, y\in \mathbb{C}$, $q,l \in \mathbb{Z}$ are arbitrary and the operator $D(s,q; y,l; x,k)$ is described as above.
\end{lemma}
 \begin{proof}
    The proof follows easily from the definitions. We substitute $\alpha'(x,j;t+s,p+q)$ and $\beta'(x,k; t+s,p+q)$ into the relation defining a Bailey pair. We want to show,
\begin{equation}
    \beta'(w,k;t+s,p+q) = M(t+s,p+q)_{w,k;x,j}\alpha'(x,j;t+s,p+q)\:, 
\end{equation}
\begin{gather}
    D(-t,-p;y,l;w,k)M(s,q)_{w,k;z,m}D(s+t,p+q;y,l;z,m) \beta(z,m; t,p) = \\
    M(s+t, p+q)_{w,k;x,j} D(s,q;y,l; x,j)\alpha(x,j; t,p) \;. 
\end{gather}
Using the reflection relation displayed
above for $D$ operators, the problem reduces to the star-triangle relation,
\begin{equation}
\begin{gathered}
       M(s,q)_{w,k; z,m}  D(s+t,q+p; y,l; z,m)  M(t,p)_{z,m; x,j} =\\
       D(t,p; y,l; w,k) M(s+t,q+p)_{w,k; x,j}  D(s,q; y,l; x,j)\:,
\end{gathered}
\end{equation}
which we have assumed to be true for $M$ and $D$ operators. 
 \end{proof}
 
 We will be constructing such $M$ and $D$ operators satisfying the star-triangle relation in virtue of the integral identities $\eqref{SU2identity}$ on $SU(2)$ and $\eqref{U1identity}$ on $U(1)$ in the following sections to construct Bailey pairs. We should mention that there is no systematic way of constructing Bailey pairs.

\subsubsection{\texorpdfstring{$SU(2)$}{SU(2)} gauge symmetry}
Let us first construct the operators that will be used to satisfy $\eqref{star triangle}$. We can basically build the operators via the characteristic properties of hypergeometric functions.
\begin{equation}
\begin{aligned}
    D(t,p; y,l;w,k)=\gamma^{(2)}(-i(t+y\pm w+\omega\rho)-i\omega_1(p\pm k+r\sigma+l);-i\omega_1r,-i\omega)\\
	\times \gamma^{(2)}(-i(t-y\pm w+\omega(1-\rho))-i\omega_1(p\pm k+r(1-\sigma)-l);-i\omega_1r,-i\omega)\\
	\times \gamma^{(2)}(-i(t-y\pm w+\omega\rho)-i\omega_2(r-(p\pm k+r\sigma+l));-i\omega_2r,-i\omega)\\
	\times \gamma^{(2)}(-i(t-y\pm w+\omega(1-\rho))-i\omega_2(r-(p\pm k+r(1-\sigma)-l));-i\omega_2r,-i\omega).
\end{aligned}
\end{equation}

Obviously, one can see that 
\begin{equation}
\begin{aligned}
D(t,p; y,l;w,k)  D(-t,-p;y,l;w,k)=1 \;,
\end{aligned}
\end{equation}
and 
\begin{equation}
\begin{aligned}
D(0,0; y,l;w,k)  =1 \;.
\end{aligned}
\end{equation}

Now we can construct the integral sum operator as follows
\begin{equation}
\begin{aligned}
M(t,p)_{z,m;x,j}=&
\frac{1}{C(t,p)}\sum_{j=0}^{[ r/2 ]} \int _{-\infty}^{\infty} 
\gamma^{(2)}(-i(-t+z\pm x)-i\omega_1(m-p\pm j);-i\omega_1r,-i\omega)      \\
	&\times  \gamma^{(2)}(-i(-t+z\pm x)-i\omega_2(r-(m-p\pm j));-i\omega_2r,-i\omega) \\
	&\times  \gamma^{(2)}(-i(-t-z\pm x)-i\omega_1(-m-p\pm j);-i\omega_1r,-i\omega) \\
	&\times \gamma^{(2)}(-i(-t-z\pm x)-i\omega_r(r-(-m-p\pm j));-i\omega_2r,-i\omega)
	 	\frac{[\mathrm{d}_j {x}]}{2r\sqrt{-\omega_1\omega_2}}\:,
\end{aligned}\label{SU2MM}
\end{equation}

where the measure of this integral operator is written as the following
\begin{equation}
\begin{aligned}
 \relax [\mathrm{d}_j {x}] = \frac{ \epsilon (j)\mathrm{d}x}{\gamma^{(2)}(\pm 2ix \pm   2i\omega_1j;-i\omega_1r,-i\omega)\gamma^{(2)}(\pm2ix-i\omega_2(r\pm2j);-i\omega_2r,-i\omega)}\:, \label{measure}
\end{aligned}
\end{equation}

and the contribution of the spin-independent function is
\begin{equation}
\begin{aligned}
C(t,p)=&\gamma^{(2)}(-i(-2t)-i\omega_1(-2p);-i\omega_1r,-i\omega)\\
&\times\gamma^{(2)}(-i(-2t)-i\omega_2(r-(-2p));-i\omega_2r,-i\omega) \;.\label{ctp}
\end{aligned}
\end{equation}
Using these operators and the identity $\eqref{star triangle}$, one can find the parameters as follows
\begin{equation}
\begin{aligned}
     {a}_{1,2} & =  -s\pm w, & {a}_3  &=s+t+y+\omega\rho\:,
     \\ 
     {a}_4 & = s+t-y+\omega(1-\rho), & {a}_{5,6} & = -t\pm x\:,
     \\
    {u}_{1,2} & =  -q\pm k, & u_3 & =q+p+l+r\sigma\:,
     \\
    {u}_4 & = q+p-l+r(1-\sigma), & {u}_{5,6}  & =  -p\pm m\:,
\end{aligned}
\end{equation}

\subsubsection{\texorpdfstring{$U(1)$}{U(1)} gauge symmetry }
For $U(1)$ gauge symmetry we need to redefine the operators such that
\begin{equation}
\begin{aligned}
    D(t,p; y,l;w,k)=\gamma^{(2)}(-i(t+y+ w+\omega\rho)-i\omega_1(p+ k+r\sigma+l);-i\omega_1r,-i\omega)\\
	\times \gamma^{(2)}(-i(t-y- w+\omega(1-\rho))-i\omega_1(p- k+r(1-\sigma)-l);-i\omega_1r,-i\omega)\\
	\times \gamma^{(2)}(-i(t-y+ w+\omega\rho)-i\omega_2(r-(p+ k+r\sigma+l));-i\omega_2r,-i\omega)\\
	\times \gamma^{(2)}(-i(t-y- w+\omega(1-\rho))-i\omega_2(r-(p- k+r(1-\sigma)-l));-i\omega_2r,-i\omega)\:.
\end{aligned}
\end{equation}
Obviously, one can again catch the reflection property of the operator such that
\begin{equation}
\begin{aligned}
D(t,p; y,l;w,k)  D(-t,-p;y,l;w,k)=1 \;,
\end{aligned}
\end{equation}
and similarly 
\begin{equation}
\begin{aligned}
 D(0,0;y,l;w,k)=1 \;.
\end{aligned}
\end{equation}
Also for the integral-sum operator, we need a slight change such that
\begin{equation}
\begin{aligned}
M(t,p)_{z,m;x,j}=&
\frac{1}{C(t,p)}\sum_{j=0}^{[ r/2 ]}\epsilon (j) \int _{-\infty}^{\infty} 
\gamma^{(2)}(-i(-t+z+x)-i\omega_1(m-p+ j);-i\omega_1r,-i\omega)      \\
	&\times  \gamma^{(2)}(-i(-t+z+ x)-i\omega_2(r-(m-p+ j));-i\omega_2r,-i\omega) \\
	&\times  \gamma^{(2)}(-i(-t-z- x)-i\omega_1(-m-p- j);-i\omega_1r,-i\omega) \\
	&\times  \gamma^{(2)}(-i(-t-z- x)-i\omega_r(r-(-m-p- j));-i\omega_2r,-i\omega)\\
	&\times 
	 	\frac{d x}{2r\sqrt{-\omega_1\omega_2}}\:,
\end{aligned}
\end{equation}
where $C(t,p)$ has the same definition in (\ref{ctp}).

As we have demonstrated in the previous section, we can re-group these multipliers at the right-hand side of the $U(1)$ gauge symmetry integral identity. By using the same methodology that was used to re-arrange the parameters of the $SU(2)$ gauge symmetry integral identity we can write the parameters of the $U(1)$ gauge symmetry integral identity as follows
\begin{equation}
\begin{aligned}
     {a}_{1},{b}_{1} & =  -s\pm w\:, & {a}_2 & =s+t+y+\omega\rho\:,
     \\ 
     {b}_2 & = s+t-y+\omega(1-\rho)\:, & {a}_{3},{b}_{3} & = -t\pm x\:,
     \\
    {u}_{1},{v}_{1} & =  -q\pm k\:, & u_2 & =q+p+l+r\sigma\:,
     \\
    {v}_2 & = q+p-l+r(1-\sigma)\:, & {u}_{3},{v}_{3} & =  -p\pm m \:.
\end{aligned}
\end{equation}

\subsection{Star-star relation}

We will now be discussing Bailey pairs generated from an initial explicit pair. Noting that $M(t,p)_{z,m;x,j}$ is an integral-sum operator acting on a sequence of functions $f_j(x)$, the relation $\eqref{definition}$ suggests to start with $\alpha(x,j;t,p) = \delta_{jn}\delta(x-u)$ where $n\in \mathbb{Z}$, $u\in \mathbb{C}$ are new parameters.

Then, $\beta(z,m;t,p)$ of the following form
\begin{equation}
 \begin{aligned}
    \beta(z,m; t,p) &= M(t,p)_{z,m;x,j}\delta_{jn}\delta(x-u) \\
    & := M(t,p;z,m; u,n) \;,
\end{aligned}   
\end{equation}
forms a Bailey pair with $\alpha(x,j;t,p)$. From here, we generate new pairs with the Bailey lemma,
\begin{align}
    \alpha(x,k; t+s;p+q) &= D(s,q; y,l;x,k)\alpha(x,k;t,p)\:, \\
    \beta(x,k; t+s; p+q) &= D(-t,-p; y,l;x,k)M(s,q)_{x,k;z,m}D(s+t,p+q; y,l;z,m)\beta(z,m;t,p) \;. 
\end{align}
The relation $\eqref{definition}$ does not give us a particularly interesting result as it yields the star-triangle relation, which we have used to prove the Bailey lemma
\begin{equation}
\begin{aligned}\label{second bailey relation}
    M&(s,q)_{w,k;z,m}D(s+t, p+q; y,l;z,m)M(t,p;z,m;u,n) \\
    &= D(t,p; y,l;w,k)M(s+t, p+q ; w,k; u,n)D(s,q;y,l,u,n)  \;.
\end{aligned}
\end{equation}
An immediate consequence of $\eqref{second bailey relation}$ is the functions $\Tilde{\alpha}(z,m;s,q)$ and $\Tilde{\beta}(w,k;s,q)$ defined by
\begin{align}
\Tilde{\alpha}(z,m;s,q) &= D(s+t, p+q; y,l;z,m)M(t,p;z,m;u,n)\:, \\
\Tilde{\beta}(w,k;s,q) &=D(t,p;y,l;w,k)M(s+t, p+q ; w,k; u,n)D(s,q;y,l,u,n) \;,
\end{align}
form a Bailey pair with respect to parameters $s \in \mathbb{C}$, $q \in \mathbb{Z}$. Applying the lemma once again, we find
\begin{align}
    \Tilde{\alpha}'(z,m;s+c,q+d) =& D(c,d;a,b;z,m)D(s+t, p+q;y,l;z,m)M(t,p;z,m;u,n)\:, \\
    \begin{split}
       \Tilde{\beta}'(x,j;s+c,q+d) =& D(-s,-q;a,b;x,j)M(c,d)_{x,j;w,k}D(s+c,q+d;a,b;w,k)\\
    & \times D(t,p;y,l;w,k)M(s+t, p+q ; w,k; u,n)D(s,q;y,l,u,n)  \;, 
    \end{split}
\end{align}
where $a, c \in \mathbb{C}$ and $b,d \in \mathbb{Z}$ are arbitrary. The relation
\begin{equation}
\Tilde{\beta}'(x,j;s+c,q+d) = M(s+c,q+d)_{x,j;z,m} \Tilde{\alpha}'(z,m;s+c,q+d)   \;,
\end{equation}
yields a non-trivial integral identity
\begin{equation}
\begin{aligned}
&M(c,d)_{x,j;w,k}D(s+c,q+d;a,b;w,k)D(t,p;y,l;w,k)M(s+t, p+q ; w,k; u,n) \\
&= D(-s,-q;y,l,u,n)D(s,q;a,b;x,j)\\
 &\text{ } \times M(s+c,q+d)_{x,j;z,m} D(c,d;a,b;z,m)D(s+t, p+q;y,l;z,m)M(t,p;z,m;u,n) \;,
\end{aligned}
\end{equation}
which can be recognized as the star-star relation.\\

\subsubsection{ \texorpdfstring{$SU(2)$}{SU(2)} gauge symmetry }
Now we need to construct this identity as an integral form. For constructing bailey pairs' operators, we will just write the same form of the operators that were used for $SU(2)$ gauge symmetry.
However, we have one more operator

\begin{equation}
\begin{aligned}
M(t,p;z,m;u,n)=&
\frac{\Delta_n^u}{C(t,p)} 
\gamma^{(2)}(-i(-t+z\pm u)-i\omega_1(m-p\pm n);-i\omega_1r,-i\omega)      \\
	&\times  \gamma^{(2)}(-i(-t+z\pm u)-i\omega_2(r-(m-p\pm n));-i\omega_2r,-i\omega) \\
	&\times  \gamma^{(2)}(-i(-t-z\pm u)-i\omega_1(-m-p\pm n);-i\omega_1r,-i\omega) \\
	&\times  \gamma^{(2)}(-i(-t-z\pm u)-i\omega_r(r-(-m-p\pm n));-i\omega_2r,-i\omega)\:,
\end{aligned}\label{SU2M}
\end{equation}

where 
\begin{equation}
\begin{aligned}
 \relax \Delta_n^u = \frac{\epsilon (n)}{2r\sqrt{-\omega_1\omega_2}}\frac{ 1}{\gamma^{(2)}(\pm 2iu \pm   2i\omega_1n;-i\omega_1r,-i\omega)\gamma^{(2)}(\pm2iu-i\omega_2(r\pm2n);-i\omega_2r,-i\omega)}\:,
\end{aligned}
\end{equation}
and $C(t,p)$ is still (\ref{ctp}).

Then one can write down the parameters as follows:
\begin{equation}
\begin{aligned}
     {a}_{1,2} & =  -c\pm x\:, & a_3 & =s+c+a+\omega\rho\:, & a_4 & = s+c-a+\omega(1-\rho)\:,
     \\ 
     {a}_5 & =  t+y+\omega\rho\:, & {a}_6 & = t-y+\omega(1-\rho)\:, & {a}_{7,8} & = -(s+t)\pm u\:, 
     \\
    {u}_{1,2} & =  -d\pm j\:, & u_3 & =q+d+b+r\sigma\:, & u_4 & = q+d-b+r(1-\sigma)\:,
     \\
    {u}_5 & =  p+l+r\sigma\:, & {u}_6 & =  p-l+r(1-\sigma)\:, & u_{7,8} & =-(p+q)\pm n\:.   
\end{aligned}
\end{equation}
\\
For the right-hand side of the integral identity, we can write the simple equality between the parameters as follows,
\begin{align}
\begin{aligned}
     \tilde{a}_i & =  a_i-s, & \tilde{u}_i & =u_i-q,  & \text{if} \;\;\; i=1,2,3,4 \:,
     \\ 
     \tilde{a}_i & =  a_i+s,  & \tilde{u}_i & = u_i+q, & \text{if} \;\;\; i=5,6,7,8\:.
\end{aligned}
\end{align}

\subsubsection{\texorpdfstring{$U(1)$}{U(1)} gauge symmetry }

For this part, all we need to do is just re-use the operators that were designed to reach the Bailey pair reconstruction of star-triangle form for the $U(1)$ gauge symmetry. 

Again, we need the following operator for the construction
\begin{equation}
\begin{aligned}
M(t,p;z,m;x,j)=&
\frac{1}{C(t,p)} \frac{\epsilon (n)}{r\sqrt{-\omega_1\omega_2}} 
\\
&\times \gamma^{(2)}(-i(-t+z+x)-i\omega_1(m-p+ j);-i\omega_1r,-i\omega)      \\
	&\times  \gamma^{(2)}(-i(-t+z+ x)-i\omega_2(r-(m-p+ j));-i\omega_2r,-i\omega) \\
	&\times  \gamma^{(2)}(-i(-t-z- x)-i\omega_1(-m-p- j);-i\omega_1r,-i\omega) \\
	&\times \gamma^{(2)}(-i(-t-z- x)-i\omega_r(r-(-m-p- j));-i\omega_2r,-i\omega)\:,
\end{aligned}
\end{equation}
where $C(t,p)$ still lives as in (\ref{ctp}).

Then one can write down the parameters as follows:
\begin{equation}
\begin{aligned}
     {a}_{1},{b}_{1} & =  -c\pm x\:, & a_2 & =s+c+a+\omega\rho\:, & b_2 & = s+c-a+\omega(1-\rho)\:,
     \\ 
     {a}_3 & =  t+y+\omega\rho\:, & {b}_3 & = t-y+\omega(1-\rho)\:, & {a}_{4},{b}_{4}& = -(s+t)\pm u\:,  
     \\
    {u}_{1},{v}_{1} & =  -d\pm j\:, & u_2 & =q+d+b+r\sigma\:, & v_2 & = q+d-b+r(1-\sigma)\:,
     \\
    {u}_3 & =  p+l+r\sigma\:, & {v}_3 & =  p-l+r(1-\sigma)\:, & {u}_{4},{v}_{4} & =-(p+q)\pm n\:,
\end{aligned}
\end{equation}
For the right-hand side of the integral identity, we can write the simple equality between the parameters as follows:
\begin{align}
\begin{aligned}
     \tilde{a}_i & =  a_i-s, & \tilde{b}_i & =  b_i-s, & \tilde{u}_i & =u_i-q, & \tilde{v}_i  & = v_i-q, & \text{if} \;\;\; i=1,2\:,
     \\ 
     \tilde{a}_i & =  a_i+s, & \tilde{b}_i & =  b_i+s, & \tilde{u}_i & = u_i+q, & \tilde{v}_i & =v_i+q, & \text{if} \;\;\; i=3,4\:.
\end{aligned}
\end{align}

\subsection{Pentagon identity}

We will now consider a different definition of Bailey pairs and its relation to a pentagon identity on $U(1)$. The discussion will be a slight generalization of \cite{Kashaev:2012cz} (see also \cite{Gahramanov:2021pgu}) with discrete parameters. 
\begin{definition}
The functions $\alpha(x,n;t,p)$ and $\beta(x,n;t,p)$ with variables $x\in \mathbb{C}$ and $n \in \mathbb{Z}$ are said to form a pentagon Bailey pair with respect to parameters $t\in \mathbb{C}$ and $p \in \mathbb{Z}$ if the following relation is satisfied
\begin{equation}
    \beta(x,n;t,p) = \sum_{m}\int_{-\infty}^{\infty} \mathrm{d}z \; \mathcal{B}(t+x-z,p+n-m; t-x+z, p-n+m) \alpha(z,m;t,p) \;,
\end{equation}
where $\mathcal{B}(z_1,u_1;z_2,u_2)$ is defined as in $\eqref{pentagon_b}$.
\end{definition}

\begin{theorem}
Suppose $\alpha(z,m;t,p)$ and $\beta(x,n; t,p)$ form a pentagon Bailey pair with respect to $t \in \mathbb{C}$ and $p \in \mathbb{Z}$. Then, the sequences of functions $\alpha'(z,m;t+s,p+q)$ and $\beta'(w,k; t+s,p+q)$, $w\in \mathbb{C}$, $k \in \mathbb{Z}$, defined by
\begin{align}
    \label{pentagon alpha} \alpha'(z,m,t+s,p+q)=& \mathcal{B}(z+t+y,m+p+l;2s,2q)\alpha(z,m;t,p)\:, \\
    \label{pentagon beta}
    \beta'(w,k;t+s,p+q)=& \sum_{n}\int_{-\infty}^{\infty} \mathrm{d}x \; \mathcal{B}(s+w-x,q+k-n; y+x,l+n) \\ \nonumber
    &\times \mathcal{B}(s+2t+y+w, q+2p+l+k;s-w+x, q-k+n)\beta(x,n;t,p)\:,
\end{align}
form a Bailey pair with respect to $t+s$ and $p+q$, where $y,s \in \mathbb{C}$, $l,q \in \mathbb{Z}$. 
\end{theorem}

\begin{proof}
We need to show the relation defining $\beta'(w,k;t+s,p+q)$ and $\alpha'(z,m,t+s,p+q)$ is satisfied,
\begin{equation}
     \beta'(w,k;t+s,p+q) = \sum_{m} \int_{-\infty}^{\infty} \mathrm{d}z \; \mathcal{B}(t+s+w-z,p+q+k-m; t+s-w+z, p+q-k+m) \alpha'(z,m;t+s,p+q) \;.
\end{equation} 
Substituting $\eqref{pentagon alpha}$ and $\eqref{pentagon beta}$ for $\alpha'(z,m,t+s,p+q)$ and $\beta'(w,k;t+s,p+q)$, we arrive at
\begin{equation}
    \begin{aligned}
  \sum_{n}\int_{-\infty}^{\infty} &\mathrm{d}x \; \mathcal{B}(s+w-x,p+k-n; y+x,l+n)
 \mathcal{B}(s+2t+y+w, q+2p+l+k;s-w+x, q-k+n)  \\
\times  \sum_{m}&\int_{-\infty}^{\infty} \mathrm{d}z \; \mathcal{B}(t+x-z,p+n-m; t-x+z, q-n+m) \alpha(z,m;t,p)\\ 
    =&  \sum_{m} \int_{-\infty}^{\infty} \mathrm{d}z \; \mathcal{B}(t+s+w-z,p+q+k-m; t+s-w+z, p+q-k+m)
    \\
    &\times \mathcal{B}(z+t+y,m+p+l;2s,2q)\alpha(z,m;t,p) \;.
    \end{aligned}   
\end{equation}
Rearrangement of some terms yields
\begin{equation} \label{pentagon integral}
\begin{aligned}
     \sum_{m}\int_{-\infty}^{\infty} &\mathrm{d}z \; \alpha(z,m;t,p) \sum_{n}\int_{-\infty}^{\infty} \mathrm{d}x \; \mathcal{B}(s+w-x,p+k-n; y+x,l+n) \\
&\times \mathcal{B}(s+2t+y+w, q+2p+l+k;s-w+x, q-k+n) 
\\
&\times \mathcal{B}(t+x-z,p+n-m; t-x+z, q-n+m) \\
=  \sum_{m} \int_{-\infty}^{\infty} &\mathrm{d}z \; \alpha(z,m;t,p)
\mathcal{B}(t+s+w-z,p+q+k-m; t+s-w+z, p+q-k+m)\\
&\times \mathcal{B}(z+t+y,m+p+l;2s,2q) \;.
\end{aligned}
\end{equation}
 With the choice of $\mathcal{B}(z_1,u_1;z_2,u_2)$ in $\eqref{pentagon_b}$, we see that $\eqref{pentagon integral}$ is indeed satisfied with the following parametrization,

\begin{align}
\begin{aligned}
     a_1&=s+w, & b_1&=y, & u_1&=q+k, & v_1&=l \:,
     \\ 
     a_2&=\omega-(2s+2t+y), & b_2&=s-w, & u_2&=-(2q+2p+l), & v_2&=q-k \:,
     \\
     a_3&=t+z, & b_3&=t-z, & u_3&=p+m, & v_3&=p-m \:.
\end{aligned}
\end{align}
and we arrive at the pentagon relation $\eqref{pentagon}$. Hence,  $\beta'(w,k;t+s,p+q)$ and $\alpha'(z,m,t+s,p+q)$ satisfy the Bailey lemma.
\end{proof}

\section{Conclusions}

We have studied novel Bailey pairs constructed from the $3d$ $\mathcal{N}=2$ dual supersymmetric gauge theories on the lens space $S_b^3/\mathbb{Z}_r$. The equality of the supersymmetric partition functions of the dual theories on the lens space leads to non-trivial hyperbolic hypergeometric integral identities. The integral identities have previously been discussed in terms of lattice spin models (Ising-like and IRF-type) in statistical mechanics and pentagon identity as a $2-3$ Pachner move \cite{Bozkurt:2020gyy,Catak:2021coz,Mullahasanoglu:2021xyf,Dede:2022ofo}. In this work, we have constructed Bailey pairs that generate these integral identities. These Bailey pair constructions allow us to study the vertex-type integrable models \cite{Gahramanov:2015cva}, knot invariants \cite{Kashaev:2012cz}, supersymmetric quiver gauge theories \cite{Brunner:2017lhb}, etc. One can use the Bailey pair construction to generate integral identities for supersymmetric dualities.  One possible future direction is to examine other supersymmetric IR dualities in this context.

In the context of integrable lattice spin models, we construct the Boltzmann weights with two types of spectral parameters\footnote{In the literature there are models with two rapidity parameters, see, e.g. \cite{Khachatryan:2012wy,Karakhanyan:2013pv}.}, and these discrete and continuous types of parameters are preserved in Bailey pair constructions.  It would be interesting to see the implications of this result. 

One can obtain the rational beta integral identities, namely the equality of supersymmetric gauge partition functions on $S^2$ by limiting $r\to\infty$ in the integral identities for dualities on $S^3_b/\mathbb{Z}_r$, see \cite{Eren:2019ibl,Sarkissian:2020pwa}. In this work, we constructed Bailey pairs for the integral identities with the balancing conditions $\sum_{i=1}^6u_i=r$ and $\sum_{i=1}^3u_i+v_i=r$. Since the limit $r \to \infty$ is problematic for the balancing conditions, it would be interesting to analyze the limiting procedure in our cases.

\section*{Acknowledgments}

The authors are grateful to Erdal Catak for the valuable discussions. Ilmar Gahramanov and Mustafa Mullahasanoglu are supported by the 1002-TUBITAK Quick Support Program under grant number 121F413. The work of Ilmar Gahramanov (Sec.4.1) is supported by the Russian Science Foundation grant number 22-72-10122. Ilmar Gahramanov is also partially supported by the Bogazici University Research Fund under grant number 20B03SUP3.

\appendix

\section{Constructing the Bailey pair for the star-triangle relation}

Let us construct the governing equation (\ref{star triangle}) for the star-triangle relation of the operators and the required identity for the Bailey pairs. We take the following steps. First, we replace the definitions of the operators on the left-hand side for a particular model. Then, the specific change of variables makes the integral part of the specific identity. Before obtaining the right side of \eqref{star triangle}, the integral identity allows us to calculate one of the integrals on the left. Finally, we call back the old variables to write the right side in proper operator form. One thing to be careful of is the spin-independent functions in the integral operator. We will mention it again when it appears in the calculation. 

Let's use the following shorthand notations in the calculations
\begin{equation}
    \begin{aligned}
    \gamma_h( z, y;\omega_1,\omega_2) =
     \gamma^{(2)}(-i z-i\omega_1y;-i\omega_1r,-i\omega)       \gamma^{(2)}(-i z-i\omega_2(r-y);-i\omega_2r,-i\omega) 
\:,
\end{aligned}
\end{equation}
where $\omega=\omega_1+\omega_2$ and
\begin{equation}
    \begin{aligned}
    \gamma_h(\pm z,\pm y;\omega_1,\omega_2) =\gamma_h( z, y;\omega_1,\omega_2)\gamma_h(- z,- y;\omega_1,\omega_2)\:,
\end{aligned}
\end{equation}

Recall that $[d_j x]$ is defined in (\ref{measure}).

Let us explicitly show how to equate the star-triangle equation \eqref{star triangle}. If one replaces the definitions of the operators, then obtains the following at the left-hand side 

\begin{equation}
\begin{aligned}
M(s,q)_{w,k; z,m}  &D(s+t,q+p; y,l; z,m)  M(t,p)_{z,m; x,j} =
\\
\frac{1}{C(s,q)C(t,p)}&\sum_{j=0}^{[ r/2 ]}\epsilon (j) \int _{-\infty}^{\infty}\frac{[d_j x]}{2r\sqrt{-\omega_1\omega_2}}
\sum_{m=0}^{[ r/2 ]}\epsilon (m) \int _{-\infty}^{\infty} \frac{[d_m z]}{2r\sqrt{-\omega_1\omega_2}}
\\  \times 
\gamma_h(&-s+ w\pm z,+ k-q\pm m;\omega_1,\omega_2) \:
\gamma_h(-s- w\pm z,- k-q\pm m;\omega_1,\omega_2)
\\  \times 
\gamma_h(&s+t+y\pm z+\omega\rho,q+p\pm m+r\sigma+l;\omega_1,\omega_2)\\
	\times  \gamma_h(&s+t-y\pm z+\omega(1-\rho),q+p\pm m+r(1-\sigma)-l;\omega_1,\omega_2) \\
	\times 	\gamma_h(&-t + z\pm x, m-p\pm j);\omega_1,\omega_2)  \:
		\gamma_h(-t - z\pm x,- m-p\pm j);\omega_1,\omega_2)\:,\label{app1} 
\end{aligned}
\end{equation}
where $C(t,p)=\gamma_h(- 2t,- 2p;\omega_1,\omega_2)$ in new notations as in (\ref{ctp}).

One can rewrite equation (\ref{app1}) under the given re-parametrization

\begin{equation}
\begin{aligned}
     {a}_{1,2} & =  -s\pm w, & {a}_3  &=s+t+y+\omega\rho\:,
     \\ 
     {a}_4 & = s+t-y+\omega(1-\rho), & {a}_{5,6} & = -t\pm x\:,
     \\
    {u}_{1,2} & =  -q\pm k, & u_3 & =q+p+l+r\sigma\:,
     \\
    {u}_4 & = q+p-l+r(1-\sigma), & {u}_{5,6}  & =  -p\pm m\:.\label{par1}
\end{aligned}
\end{equation}
Then, (\ref{app1}) becomes

\begin{equation}
\begin{aligned}
=\sum_{j=0}^{[ r/2 ]}\epsilon (j) \int _{-\infty}^{\infty}&\frac{[d_j x]}{2r\sqrt{-\omega_1\omega_2}}
\sum_{m=0}^{[ r/2 ]}\epsilon (m) \int _{-\infty}^{\infty} \frac{[d_m z]}{2r\sqrt{-\omega_1\omega_2}}
\\ \times
&\frac{\prod_{i=1}^6 \gamma_h(a_i \pm z,u_i \pm m;\omega_1,\omega_2)}{\gamma_h(a_1+a_2,u_1+u_2;\omega_1,\omega_2)\gamma_h(a_5+a_6,u_5+u_6;\omega_1,\omega_2)} \;. 
\end{aligned}
\end{equation}
One can integrate with respect to $z$ and sum on $m$ by using the integral identity (\ref{SU2identity}) and the result is

\begin{equation}
\begin{aligned}
   = \sum_{j=0}^{[ r/2 ]}\epsilon (j) &\int _{-\infty}^{\infty} 
    \frac{ \prod_{1\leq i<j\leq 6} \gamma_h(a_i + a_j,u_i + u_j;\omega_1,\omega_2) }{\gamma_h(a_1+a_2,u_1+u_2;\omega_1,\omega_2)\gamma_h(a_5+a_6,u_5+u_6;\omega_1,\omega_2)}
     \frac{[d_m x]}{2r\sqrt{-\omega_1\omega_2}}\;.\label{app2} 
\end{aligned}
\end{equation}

If one replaces the parameters (\ref{par1}) with their values found before, it is easy to see that (\ref{app2}) turns into
\begin{equation}
\begin{aligned}
=D(t,p;y,l;w,k)  M(s+t,q+p)_{w,k; x,j}  D(s,q;y,l;x,j) \:.
\end{aligned}
\end{equation}

This is the right-hand side of the star-triangle relation (\ref{star triangle}).

\bibliographystyle{utphys}
\bibliography{refYBE}

\providecommand{\href}[2]{#2}\begingroup\raggedright\begin{thebibliography}{100}

\bibitem{Gang:2019juz}
D.~Gang, ``{Chern-Simons Theory on $L(p,q)$ Lens Spaces and Localization},''
  \href{http://dx.doi.org/10.3938/jkps.74.1119}{{\em J. Korean Phys. Soc.} {\bf
  74} (2019) no.~12, 1119--1128}, \href{http://arxiv.org/abs/0912.4664}{{\tt
  arXiv:0912.4664 [hep-th]}}.

\bibitem{Benini:2011nc}
F.~Benini, T.~Nishioka, and M.~Yamazaki, ``{4d Index to 3d Index and 2d
  TQFT},'' \href{http://dx.doi.org/10.1103/PhysRevD.86.065015}{{\em Phys. Rev.}
  {\bf D86} (2012)  065015},
\href{http://arxiv.org/abs/1109.0283}{{\tt arXiv:1109.0283 [hep-th]}}.

\bibitem{Imamura:2012rq}
Y.~Imamura and D.~Yokoyama, ``{$S^3/Z_n$ partition function and dualities},''
  \href{http://dx.doi.org/10.1007/JHEP11(2012)122}{{\em JHEP} {\bf 11} (2012)
  122},
\href{http://arxiv.org/abs/1208.1404}{{\tt arXiv:1208.1404 [hep-th]}}.

\bibitem{Imamura:2013qxa}
Y.~Imamura, H.~Matsuno, and D.~Yokoyama, ``{Factorization of the
  $S^3/\mathbb{Z}_n$ partition function},''
  \href{http://dx.doi.org/10.1103/PhysRevD.89.085003}{{\em Phys. Rev.} {\bf
  D89} (2014) no.~8, 085003},
\href{http://arxiv.org/abs/1311.2371}{{\tt arXiv:1311.2371 [hep-th]}}.

\bibitem{Nieri:2015yia}
F.~Nieri and S.~Pasquetti, ``{Factorisation and holomorphic blocks in 4d},''
  \href{http://dx.doi.org/10.1007/JHEP11(2015)155}{{\em JHEP} {\bf 11} (2015)
  155},
\href{http://arxiv.org/abs/1507.00261}{{\tt arXiv:1507.00261 [hep-th]}}.

\bibitem{Alday:2012au}
L.~F. Alday, M.~Fluder, and J.~Sparks, ``{The Large N limit of M2-branes on
  Lens spaces},'' \href{http://dx.doi.org/10.1007/JHEP10(2012)057}{{\em JHEP}
  {\bf 10} (2012)  057},
\href{http://arxiv.org/abs/1204.1280}{{\tt arXiv:1204.1280 [hep-th]}}.

\bibitem{Yamazaki:2013fva}
M.~Yamazaki, ``{Four-dimensional superconformal index reloaded},''
  \href{http://dx.doi.org/10.1007/s11232-013-0012-6}{{\em Theor. Math. Phys.}
  {\bf 174} (2013)  154--166}.
[Teor. Mat. Fiz.174,177(2013)].

\bibitem{Honda:2016vmv}
M.~Honda, ``{How to resum perturbative series in 3d N=2 Chern-Simons matter
  theories},'' \href{http://dx.doi.org/10.1103/PhysRevD.94.025039}{{\em Phys.
  Rev. D} {\bf 94} (2016) no.~2, 025039},
  \href{http://arxiv.org/abs/1604.08653}{{\tt arXiv:1604.08653 [hep-th]}}.

\bibitem{Nedelin:2016gwu}
A.~Nedelin, F.~Nieri, and M.~Zabzine, ``{$q$-Virasoro modular double and 3d
  partition functions},''
  \href{http://dx.doi.org/10.1007/s00220-017-2882-1}{{\em Commun. Math. Phys.}
  {\bf 353} (2017) no.~3, 1059--1102},
  \href{http://arxiv.org/abs/1605.07029}{{\tt arXiv:1605.07029 [hep-th]}}.

\bibitem{Spiridonov:2009za}
V.~P. Spiridonov and G.~S. Vartanov, ``{Elliptic Hypergeometry of
  Supersymmetric Dualities},''
  \href{http://dx.doi.org/10.1007/s00220-011-1218-9}{{\em Commun. Math. Phys.}
  {\bf 304} (2011)  797--874}, \href{http://arxiv.org/abs/0910.5944}{{\tt
  arXiv:0910.5944 [hep-th]}}.

\bibitem{Spiridonov2014}
V.~P. Spiridonov and G.~S. Vartanov, ``Elliptic hypergeometry of supersymmetric
  dualities ii. orthogonal groups, knots, and vortices,''
  \href{http://dx.doi.org/10.1007/s00220-013-1861-4}{{\em Commun. Math. Phys.}
  {\bf 325} (2014) no.~2, 421--486},
\href{http://arxiv.org/abs/1107.5788}{{\tt arXiv:1107.5788 [hep-th]}}.

\bibitem{Krattenthaler:2011da}
C.~Krattenthaler, V.~P. Spiridonov, and G.~S. Vartanov, ``{Superconformal
  indices of three-dimensional theories related by mirror symmetry},''
  \href{http://dx.doi.org/10.1007/JHEP06(2011)008}{{\em JHEP} {\bf 06} (2011)
  008},
\href{http://arxiv.org/abs/1103.4075}{{\tt arXiv:1103.4075 [hep-th]}}.

\bibitem{Gahramanov:gka}
I.~B. Gahramanov and G.~S. Vartanov,
  \href{http://dx.doi.org/10.1142/9789814449243_0076}{``{Superconformal indices
  and partition functions for supersymmetric field theories},''} in {\em
  {XVIIth Intern. Cong. Math. Phys. 695-703 (2013)}}.
\newblock 2013.
\newblock
\href{http://arxiv.org/abs/1310.8507}{{\tt arXiv:1310.8507 [hep-th]}}.
\newblock

\bibitem{Dolan:2011rp}
F.~A.~H. Dolan, V.~P. Spiridonov, and G.~S. Vartanov, ``{From 4d superconformal
  indices to 3d partition functions},''
  \href{http://dx.doi.org/10.1016/j.physletb.2011.09.007}{{\em Phys. Lett.}
  {\bf B704} (2011)  234--241},
\href{http://arxiv.org/abs/1104.1787}{{\tt arXiv:1104.1787 [hep-th]}}.

\bibitem{Gahramanov:2015tta}
I.~Gahramanov, ``{Mathematical structures behind supersymmetric dualities},''
  \href{http://dx.doi.org/10.5817/AM2015-5-273}{{\em Archivum Math.} {\bf 51}
  (2015)  273--286},
\href{http://arxiv.org/abs/1505.05656}{{\tt arXiv:1505.05656 [math-ph]}}.

\bibitem{Tachikawa:2017byo}
Y.~Tachikawa, ``{On 'categories' of quantum field theories},'' in {\em
  {International Congress of Mathematicians}}.
\newblock 12, 2017.
\newblock \href{http://arxiv.org/abs/1712.09456}{{\tt arXiv:1712.09456
  [math-ph]}}.

\bibitem{Spiridonov:2019kto}
V.~P. Spiridonov, ``{Superconformal Indices, Seiberg Dualities and Special
  Functions},'' \href{http://dx.doi.org/10.1134/S1063779620040681}{{\em Phys.
  Part. Nucl.} {\bf 51} (2020) no.~4, 508--513},
  \href{http://arxiv.org/abs/1912.11514}{{\tt arXiv:1912.11514 [hep-th]}}.

\bibitem{Gahramanov:2022qge}
I.~Gahramanov, ``{Integrability from supersymmetric duality: a short review},''
  \href{http://arxiv.org/abs/2201.00351}{{\tt arXiv:2201.00351 [hep-th]}}.

\bibitem{Spiridonov:2010em}
V.~P. Spiridonov, ``{Elliptic beta integrals and solvable models of statistical
  mechanics},'' {\em Contemp. Math.} {\bf 563} (2012)  181--211,
\href{http://arxiv.org/abs/1011.3798}{{\tt arXiv:1011.3798 [hep-th]}}.

\bibitem{Yamazaki:2012cp}
M.~Yamazaki, ``{Quivers, YBE and 3-manifolds},''
  \href{http://dx.doi.org/10.1007/JHEP05(2012)147}{{\em JHEP} {\bf 05} (2012)
  147},
\href{http://arxiv.org/abs/1203.5784}{{\tt arXiv:1203.5784 [hep-th]}}.

\bibitem{Kels:2015bda}
A.~P. Kels, ``{New solutions of the star–triangle relation with discrete and
  continuous spin variables},''
  \href{http://dx.doi.org/10.1088/1751-8113/48/43/435201}{{\em J. Phys.} {\bf
  A48} (2015) no.~43, 435201},
\href{http://arxiv.org/abs/1504.07074}{{\tt arXiv:1504.07074 [math-ph]}}.

\bibitem{Yagi:2015lha}
J.~Yagi, ``{Quiver gauge theories and integrable lattice models},''
  \href{http://dx.doi.org/10.1007/JHEP10(2015)065}{{\em JHEP} {\bf 10} (2015)
  065},
\href{http://arxiv.org/abs/1504.04055}{{\tt arXiv:1504.04055 [hep-th]}}.

\bibitem{Gahramanov:2015cva}
I.~Gahramanov and V.~P. Spiridonov, ``{The star-triangle relation and 3d
  superconformal indices},''
  \href{http://dx.doi.org/10.1007/JHEP08(2015)040}{{\em JHEP} {\bf 08} (2015)
  040},
\href{http://arxiv.org/abs/1505.00765}{{\tt arXiv:1505.00765 [hep-th]}}.

\bibitem{Gahramanov:2016ilb}
I.~Gahramanov and A.~P. Kels, ``{The star-triangle relation, lens partition
  function, and hypergeometric sum/integrals},''
  \href{http://dx.doi.org/10.1007/JHEP02(2017)040}{{\em JHEP} {\bf 02} (2017)
  040},
\href{http://arxiv.org/abs/1610.09229}{{\tt arXiv:1610.09229 [math-ph]}}.

\bibitem{Sarkissian:2018ppc}
G.~Sarkissian and V.~P. Spiridonov, ``{From rarefied elliptic beta integral to
  parafermionic star-triangle relation},''
  \href{http://dx.doi.org/10.1007/JHEP10(2018)097}{{\em JHEP} {\bf 10} (2018)
  097},
\href{http://arxiv.org/abs/1809.00493}{{\tt arXiv:1809.00493 [hep-th]}}.

\bibitem{Eren:2019ibl}
E.~Eren, I.~Gahramanov, S.~Jafarzade, and G.~Mogol, ``{Gamma function solutions
  to the star-triangle equation},''
  \href{http://dx.doi.org/10.1016/j.nuclphysb.2020.115283}{{\em Nucl. Phys. B}
  {\bf 963} (2021)  115283}, \href{http://arxiv.org/abs/1912.12271}{{\tt
  arXiv:1912.12271 [math-ph]}}.

\bibitem{de-la-Cruz-Moreno:2020xop}
J.~de-la Cruz-Moreno and H.~Garc\'\i{}a-Compe\'an, ``{Star-triangle type
  relations from $2d$ $\mathcal{N}=(0,2)$ $USp(2N)$ dualities},''
  \href{http://dx.doi.org/10.1007/JHEP01(2021)023}{{\em JHEP} {\bf 01} (2021)
  023}, \href{http://arxiv.org/abs/2008.02419}{{\tt arXiv:2008.02419
  [hep-th]}}.

\bibitem{Bozkurt:2020gyy}
D.~N. Bozkurt, I.~Gahramanov, and M.~Mullahasanoglu, ``{Lens partition
  function, pentagon identity, and star-triangle relation},''
  \href{http://dx.doi.org/10.1103/PhysRevD.103.126013}{{\em Phys. Rev. D} {\bf
  103} (2021) no.~12, 126013}, \href{http://arxiv.org/abs/2009.14198}{{\tt
  arXiv:2009.14198 [hep-th]}}.

\bibitem{Yamazaki:2015voa}
M.~Yamazaki and W.~Yan, ``{Integrability from 2d ${\mathcal{N}}=(2,2)$
  dualities},'' \href{http://dx.doi.org/10.1088/1751-8113/48/39/394001}{{\em J.
  Phys.} {\bf A48} (2015)  394001},
\href{http://arxiv.org/abs/1504.05540}{{\tt arXiv:1504.05540 [hep-th]}}.

\bibitem{Kels:2017toi}
A.~P. Kels and M.~Yamazaki, ``{Elliptic hypergeometric sum/integral
  transformations and supersymmetric lens index},''
  \href{http://dx.doi.org/10.3842/SIGMA.2018.013}{{\em SIGMA} {\bf 14} (2018)
  013}, \href{http://arxiv.org/abs/1704.03159}{{\tt arXiv:1704.03159
  [math-ph]}}.

\bibitem{Catak:2021coz}
E.~Catak, I.~Gahramanov, and M.~Mullahasanoglu, ``{Hyperbolic and trigonometric
  hypergeometric solutions to the star-star equation},''
  \href{http://dx.doi.org/10.1140/epjc/s10052-022-10661-x}{{\em Eur. Phys. J.
  C} {\bf 82} (2022)  789}, \href{http://arxiv.org/abs/2107.06880}{{\tt
  arXiv:2107.06880 [hep-th]}}.

\bibitem{Mullahasanoglu:2021xyf}
M.~Mullahasanoglu and N.~Tas, ``{Lens Partition Functions and Integrability
  Properties},'' \href{http://arxiv.org/abs/2112.15161}{{\tt arXiv:2112.15161
  [hep-th]}}.

\bibitem{Kashaev:2012cz}
R.~Kashaev, F.~Luo, and G.~Vartanov, ``{A TQFT of Turaev\textendash{}Viro Type
  on Shaped Triangulations},''
  \href{http://dx.doi.org/10.1007/s00023-015-0427-8}{{\em Annales Henri
  Poincare} {\bf 17} (2016) no.~5, 1109--1143},
  \href{http://arxiv.org/abs/1210.8393}{{\tt arXiv:1210.8393 [math.QA]}}.

\bibitem{Kashaev:2014rea}
R.~M. Kashaev, ``{Beta pentagon relations},''
  \href{http://dx.doi.org/10.1007/s11232-014-0208-4}{{\em Theor. Math. Phys.}
  {\bf 181} (2014) no.~1, 1194--1205},
  \href{http://arxiv.org/abs/1403.1298}{{\tt arXiv:1403.1298 [math-ph]}}.

\bibitem{kashaev2014euler}
R.~Kashaev, ``Euler’s beta function and pentagon relations,''
  \href{http://dx.doi.org/https://doi.org/10.1007/s40306-014-0080-1}{{\em Acta
  Mathematica Vietnamica} {\bf 39} (2014) no.~4, 561--566}.

\bibitem{Gahramanov:2013rda}
I.~Gahramanov and H.~Rosengren, ``{A new pentagon identity for the tetrahedron
  index},'' \href{http://dx.doi.org/10.1007/JHEP11(2013)128}{{\em JHEP} {\bf
  11} (2013)  128},
\href{http://arxiv.org/abs/1309.2195}{{\tt arXiv:1309.2195 [hep-th]}}.

\bibitem{Gahramanov:2014ona}
I.~Gahramanov and H.~Rosengren, ``{Integral pentagon relations for 3d
  superconformal indices},'' \href{http://arxiv.org/abs/1412.2926}{{\tt
  arXiv:1412.2926 [hep-th]}}.
[Proc. Symp. Pure Math.93,165(2016)].

\bibitem{Gahramanov:2016wxi}
I.~Gahramanov and H.~Rosengren, ``{Basic hypergeometry of supersymmetric
  dualities},'' \href{http://dx.doi.org/10.1016/j.nuclphysb.2016.10.004}{{\em
  Nucl. Phys.} {\bf B913} (2016)  747--768},
\href{http://arxiv.org/abs/1606.08185}{{\tt arXiv:1606.08185 [hep-th]}}.

\bibitem{Bozkurt:2018xno}
D.~N. Bozkurt and I.~Gahramanov, ``{Pentagon identities arising in
  supersymmetric gauge theory computations},''
  \href{http://dx.doi.org/10.1134/S0040577919020028}{{\em Teor. Mat. Fiz.} {\bf
  198} (2019) no.~2, 215--224},
\href{http://arxiv.org/abs/1803.00855}{{\tt arXiv:1803.00855 [math-ph]}}.

\bibitem{Jafarzade:2018yei}
S.~Jafarzade, ``{New Pentagon Identities Revisited},''
  \href{http://dx.doi.org/10.1088/1742-6596/1194/1/012054}{{\em J. Phys. Conf.
  Ser.} {\bf 1194} (2019) no.~1, 012054},
\href{http://arxiv.org/abs/1812.01325}{{\tt arXiv:1812.01325 [math-ph]}}.

\bibitem{Dede:2022ofo}
M.~Dede, ``{A comment on the solutions of the generalized Faddeev-Volkov
  model},'' \href{http://arxiv.org/abs/2206.14271}{{\tt arXiv:2206.14271
  [math-ph]}}.

\bibitem{spiridonov2004bailey}
V.~P. Spiridonov, ``A bailey tree for integrals,'' {\em Theoretical and
  mathematical physics} {\bf 139} (2004) no.~1, 536--541.

\bibitem{Brunner:2017lhb}
F.~Br\"unner and V.~P. Spiridonov, ``{4d $\mathcal{N}=1$ quiver gauge theories
  and the $\mathrm{A_n}$ Bailey lemma},''
  \href{http://dx.doi.org/10.1007/JHEP03(2018)105}{{\em JHEP} {\bf 03} (2018)
  105}, \href{http://arxiv.org/abs/1712.07018}{{\tt arXiv:1712.07018
  [hep-th]}}.

\bibitem{Spiridonov_2019}
V.~P. Spiridonov, ``The rarefied elliptic bailey lemma and the
  yang{\textendash}baxter equation,''
  \href{http://dx.doi.org/10.1088/1751-8121/ab3358}{{\em Journal of Physics A:
  Mathematical and Theoretical} {\bf 52} (2019) no.~35, 355201},
  \href{http://arxiv.org/abs/1904.12046}{{\tt arXiv:1904.12046 [math-ph]}}.

\bibitem{Gahramanov:2021pgu}
I.~Gahramanov and O.~E. Kaluc, ``{Bailey pairs for the q-hypergeometric
  integral pentagon identity},'' \href{http://arxiv.org/abs/2111.14793}{{\tt
  arXiv:2111.14793 [math-ph]}}.

\bibitem{Hadasz:2013bwa}
L.~Hadasz, M.~Pawelkiewicz, and V.~Schomerus, ``{Self-dual Continuous Series of
  Representations for $\mathcal{U}_q(sl(2))$ and $\mathcal{U}_q(osp(1|2))$},''
  \href{http://dx.doi.org/10.1007/JHEP10(2014)091}{{\em JHEP} {\bf 10} (2014)
  091}, \href{http://arxiv.org/abs/1305.4596}{{\tt arXiv:1305.4596 [hep-th]}}.

\bibitem{Fan:2021bwt}
Y.~Fan and T.~G. Mertens, ``{From quantum groups to Liouville and dilaton
  quantum gravity},'' \href{http://dx.doi.org/10.1007/JHEP05(2022)092}{{\em
  JHEP} {\bf 05} (2022)  092}, \href{http://arxiv.org/abs/2109.07770}{{\tt
  arXiv:2109.07770 [hep-th]}}.

\bibitem{Apresyan:2022erh}
E.~Apresyan, G.~Sarkissian, and V.~P. Spiridonov, ``{A parafermionic
  hypergeometric function and supersymmetric 6j-symbols},''
  \href{http://arxiv.org/abs/2205.10276}{{\tt arXiv:2205.10276 [hep-th]}}.

\bibitem{Gahramanov:2017ysd}
I.~Gahramanov and S.~Jafarzade, ``{Integrable lattice spin models from
  supersymmetric dualities},''
  \href{http://dx.doi.org/10.1134/S1547477118060079}{{\em Phys. Part. Nucl.
  Lett.} {\bf 15} (2018) no.~6, 650--667},
\href{http://arxiv.org/abs/1712.09651}{{\tt arXiv:1712.09651 [math-ph]}}.

\bibitem{Yamazaki:2018xbx}
M.~Yamazaki, ``{Integrability As Duality: The Gauge/YBE Correspondence},''
  \href{http://dx.doi.org/10.1016/j.physrep.2020.01.006}{{\em Phys. Rept.} {\bf
  859} (2020)  1--20}, \href{http://arxiv.org/abs/1808.04374}{{\tt
  arXiv:1808.04374 [hep-th]}}.

\bibitem{Bazhanov:2007mh}
V.~V. Bazhanov, V.~V. Mangazeev, and S.~M. Sergeev, ``{Faddeev-Volkov solution
  of the Yang-Baxter equation and discrete conformal symmetry},''
  \href{http://dx.doi.org/10.1016/j.nuclphysb.2007.05.013}{{\em Nucl. Phys. B}
  {\bf 784} (2007)  234--258},
\href{http://arxiv.org/abs/hep-th/0703041}{{\tt arXiv:hep-th/0703041
  [hep-th]}}.

\bibitem{Bazhanov:2007vg}
V.~V. Bazhanov, V.~V. Mangazeev, and S.~M. Sergeev, ``{Exact solution of the
  Faddeev-Volkov model},''
  \href{http://dx.doi.org/10.1016/j.physleta.2007.10.053}{{\em Phys. Lett.}
  {\bf A372} (2008)  1547--1550},
\href{http://arxiv.org/abs/0706.3077}{{\tt arXiv:0706.3077
  [cond-mat.stat-mech]}}.

\bibitem{Hama:2011ea}
N.~Hama, K.~Hosomichi, and S.~Lee, ``{SUSY Gauge Theories on Squashed
  Three-Spheres},'' \href{http://dx.doi.org/10.1007/JHEP05(2011)014}{{\em JHEP}
  {\bf 05} (2011)  014},
\href{http://arxiv.org/abs/1102.4716}{{\tt arXiv:1102.4716 [hep-th]}}.

\bibitem{baxter:1997ssr}
R.~J. Baxter, ``Star-triangle and star-star relations in statistical
  mechanics,'' \href{http://dx.doi.org/10.1142/S0217979297000058}{{\em
  International Journal of Modern Physics B} {\bf 11} (1997) no.~01n02,
  27--37}.

\bibitem{Seiberg:1994pq}
N.~Seiberg, ``{Electric - magnetic duality in supersymmetric nonAbelian gauge
  theories},'' \href{http://dx.doi.org/10.1016/0550-3213(94)00023-8}{{\em Nucl.
  Phys.} {\bf B435} (1995)  129--146},
\href{http://arxiv.org/abs/hep-th/9411149}{{\tt arXiv:hep-th/9411149
  [hep-th]}}.

\bibitem{Intriligator:1996ex}
K.~A. Intriligator and N.~Seiberg, ``{Mirror symmetry in three-dimensional
  gauge theories},'' \href{http://dx.doi.org/10.1016/0370-2693(96)01088-X}{{\em
  Phys. Lett. B} {\bf 387} (1996)  513--519},
  \href{http://arxiv.org/abs/hep-th/9607207}{{\tt arXiv:hep-th/9607207}}.

\bibitem{Aharony:1997bx}
O.~Aharony, A.~Hanany, K.~A. Intriligator, N.~Seiberg, and M.~J. Strassler,
  ``{Aspects of N=2 supersymmetric gauge theories in three-dimensions},''
  \href{http://dx.doi.org/10.1016/S0550-3213(97)00323-4}{{\em Nucl. Phys. B}
  {\bf 499} (1997)  67--99}, \href{http://arxiv.org/abs/hep-th/9703110}{{\tt
  arXiv:hep-th/9703110}}.

\bibitem{Kapustin:2011jm}
A.~Kapustin and B.~Willett, ``{Generalized Superconformal Index for Three
  Dimensional Field Theories},''
\href{http://arxiv.org/abs/1106.2484}{{\tt arXiv:1106.2484 [hep-th]}}.

\bibitem{Kapustin:2010xq}
A.~Kapustin, B.~Willett, and I.~Yaakov, ``{Nonperturbative Tests of
  Three-Dimensional Dualities},''
  \href{http://dx.doi.org/10.1007/JHEP10(2010)013}{{\em JHEP} {\bf 10} (2010)
  013},
\href{http://arxiv.org/abs/1003.5694}{{\tt arXiv:1003.5694 [hep-th]}}.

\bibitem{Amariti:2015vwa}
A.~Amariti, ``{Integral identities for 3d dualities with SP(2N) gauge
  groups},''
\href{http://arxiv.org/abs/1509.02199}{{\tt arXiv:1509.02199 [hep-th]}}.

\bibitem{Kharchev:2001rs}
S.~Kharchev, D.~Lebedev, and M.~Semenov-Tian-Shansky, ``{Unitary
  representations of U(q) (sl(2, R)), the modular double, and the multiparticle
  q deformed Toda chains},''
  \href{http://dx.doi.org/10.1007/s002200100592}{{\em Commun. Math. Phys.} {\bf
  225} (2002)  573--609},
\href{http://arxiv.org/abs/hep-th/0102180}{{\tt arXiv:hep-th/0102180
  [hep-th]}}.

\bibitem{Ponsot:2000mt}
B.~Ponsot and J.~Teschner, ``{Clebsch-Gordan and Racah-Wigner coefficients for
  a continuous series of representations of U(q)(sl(2,R))},''
  \href{http://dx.doi.org/10.1007/PL00005590}{{\em Commun. Math. Phys.} {\bf
  224} (2001)  613--655},
\href{http://arxiv.org/abs/math/0007097}{{\tt arXiv:math/0007097 [math-qa]}}.

\bibitem{Faddeev:1999fe}
L.~D. Faddeev, ``{Modular double of quantum group},'' {\em Math. Phys. Stud.}
  {\bf 21} (2000)  149--156, \href{http://arxiv.org/abs/math/9912078}{{\tt
  arXiv:math/9912078 [math-qa]}}.
[,149(1999)].

\bibitem{Faddeev:2000if}
L.~D. Faddeev, R.~M. Kashaev, and A.~{\relax Yu}. Volkov, ``{Strongly coupled
  quantum discrete Liouville theory. 1. Algebraic approach and duality},''
  \href{http://dx.doi.org/10.1007/s002200100412}{{\em Commun. Math. Phys.} {\bf
  219} (2001)  199--219},
\href{http://arxiv.org/abs/hep-th/0006156}{{\tt arXiv:hep-th/0006156
  [hep-th]}}.

\bibitem{Volkov:2005zrq}
A.~Y. Volkov, ``{Noncommutative hypergeometry},''
  \href{http://dx.doi.org/10.1007/s00220-005-1342-5}{{\em Commun. Math. Phys.}
  {\bf 258} (2005)  257--273}, \href{http://arxiv.org/abs/math/0312084}{{\tt
  arXiv:math/0312084}}.

\bibitem{Bazhanov:2010kz}
V.~V. Bazhanov and S.~M. Sergeev, ``{A Master solution of the quantum
  Yang-Baxter equation and classical discrete integrable equations},''
  \href{http://dx.doi.org/10.4310/ATMP.2012.v16.n1.a3}{{\em Adv. Theor. Math.
  Phys.} {\bf 16} (2012) no.~1, 65--95},
\href{http://arxiv.org/abs/1006.0651}{{\tt arXiv:1006.0651 [math-ph]}}.

\bibitem{Spiridonov-essays}
V.~P. Spiridonov, ``Essays on the theory of elliptic hypergeometric
  functions,'' \href{http://dx.doi.org/10.1070/RM2008v063n03ABEH004533}{{\em
  Russian Mathematical Surveys} {\bf 63} (2008) no.~3, 405},
  \href{http://arxiv.org/abs/0805.3135}{{\tt arXiv:0805.3135 [math.CA]}}.

\bibitem{Ruijsenaars:1997:FOA}
S.~N.~M. Ruijsenaars, ``First order analytic difference equations and
  integrable quantum systems,'' \href{http://dx.doi.org/10.1063/1.531809}{{\em
  J. Math. Phys.} {\bf 38} (1997) no.~2, 1069--1146}.

\bibitem{van2007hyperbolic}
F.~van~de Bult {\em et al.}, ``Hyperbolic hypergeometric functions,'' {\em
  Ph.D. Thesis, University of Amsterdam, Amsterdam Netherlands} (2007)  .

\bibitem{Andersen:2014aoa}
J.~E. Andersen and R.~Kashaev, ``{Complex Quantum Chern-Simons},''
  \href{http://arxiv.org/abs/1409.1208}{{\tt arXiv:1409.1208 [math.QA]}}.

\bibitem{Faddeev:1995nb}
L.~Faddeev, ``{Discrete Heisenberg-Weyl group and modular group},''
  \href{http://dx.doi.org/10.1007/BF01872779}{{\em Lett. Math. Phys.} {\bf 34}
  (1995)  249--254}, \href{http://arxiv.org/abs/hep-th/9504111}{{\tt
  arXiv:hep-th/9504111}}.

\bibitem{woronowicz2000quantum}
S.~Woronowicz, ``Quantum exponential function,''
  \href{http://dx.doi.org/10.1142/S0129055X00000344}{{\em Reviews in
  Mathematical Physics} {\bf 12} (2000) no.~06, 873--920}.

\bibitem{hikami2001hyperbolic}
K.~Hikami, ``Hyperbolic structure arising from a knot invariant,''
  \href{http://dx.doi.org/10.1142/s0217751x0100444x}{{\em International Journal
  of Modern Physics A} {\bf 16} (2001) no.~19, 3309--3333}.

\bibitem{Hikami2007}
K.~Hikami, ``Generalized volume conjecture and the a-polynomials: The
  neumann–zagier potential function as a classical limit of the partition
  function,'' \href{http://dx.doi.org/10.1016/j.geomphys.2007.03.008}{{\em
  Journal of Geometry and Physics} {\bf 57} (2007) no.~9, 1895–1940}.

\bibitem{hikami2014braiding}
K.~Hikami and R.~Inoue, ``Braiding operator via quantum cluster algebra,''
  \href{http://dx.doi.org/10.1088/1751-8113/47/47/474006}{{\em J. Phys A} {\bf
  47} (2014) no.~47, 474006}.

\bibitem{Chan:2017qnw}
C.-T. Chan, A.~Mironov, A.~Morozov, and A.~Sleptsov, ``{Orthogonal Polynomials
  in Mathematical Physics},''
  \href{http://dx.doi.org/10.1142/9789813233867\_0011}{{\em Rev.Math.Phys} {\bf
  30} (2018)  1840005}, \href{http://arxiv.org/abs/1712.03155}{{\tt
  arXiv:1712.03155 [hep-th]}}.

\bibitem{Teschner:2012em}
J.~Teschner and G.~Vartanov, ``{6j symbols for the modular double, quantum
  hyperbolic geometry, and supersymmetric gauge theories},''
  \href{http://dx.doi.org/10.1007/s11005-014-0684-3}{{\em Lett. Math. Phys.}
  {\bf 104} (2014)  527--551},
\href{http://arxiv.org/abs/1202.4698}{{\tt arXiv:1202.4698 [hep-th]}}.

\bibitem{Baxter:1982zz}
R.~J. Baxter, {\em Exactly {S}olved {M}odels in {S}tatistical {M}echanics}.
\newblock Academic, London,
1982.
\newblock

\bibitem{Spiridonov:2008zr}
V.~P. Spiridonov and G.~S. Vartanov, ``{Superconformal indices for N = 1
  theories with multiple duals},''
  \href{http://dx.doi.org/10.1016/j.nuclphysb.2009.08.022}{{\em Nucl. Phys. B}
  {\bf 824} (2010)  192--216}, \href{http://arxiv.org/abs/0811.1909}{{\tt
  arXiv:0811.1909 [hep-th]}}.

\bibitem{Dimofte:2012pd}
T.~Dimofte and D.~Gaiotto, ``{An E7 Surprise},''
  \href{http://dx.doi.org/10.1007/JHEP10(2012)129}{{\em JHEP} {\bf 10} (2012)
  129}, \href{http://arxiv.org/abs/1209.1404}{{\tt arXiv:1209.1404 [hep-th]}}.

\bibitem{Bazhanov:2011mz}
V.~V. Bazhanov and S.~M. Sergeev, ``{Elliptic gamma-function and multi-spin
  solutions of the Yang-Baxter equation},''
  \href{http://dx.doi.org/10.1016/j.nuclphysb.2011.10.032}{{\em Nucl. Phys.}
  {\bf B856} (2012)  475--496},
\href{http://arxiv.org/abs/1106.5874}{{\tt arXiv:1106.5874 [math-ph]}}.

\bibitem{Bazhanov:2013bh}
V.~V. Bazhanov, A.~P. Kels, and S.~M. Sergeev, ``{Comment on star-star
  relations in statistical mechanics and elliptic gamma-function identities},''
  \href{http://dx.doi.org/10.1088/1751-8113/46/15/152001}{{\em J. Phys.} {\bf
  A46} (2013)  152001},
\href{http://arxiv.org/abs/1301.5775}{{\tt arXiv:1301.5775 [math-ph]}}.

\bibitem{pachner1991pl}
U.~Pachner, ``P.l. homeomorphic manifolds are equivalent by elementary
  shellings,'' \href{http://dx.doi.org/10.1016/S0195-6698(13)80080-7}{{\em
  European journal of Combinatorics} {\bf 12} (1991) no.~2, 129--145}.

\bibitem{Dimofte:2011ju}
T.~Dimofte, D.~Gaiotto, and S.~Gukov, ``{Gauge Theories Labelled by
  Three-Manifolds},'' \href{http://dx.doi.org/10.1007/s00220-013-1863-2}{{\em
  Commun. Math. Phys.} {\bf 325} (2014)  367--419},
  \href{http://arxiv.org/abs/1108.4389}{{\tt arXiv:1108.4389 [hep-th]}}.

\bibitem{Dimofte:2011py}
T.~Dimofte, D.~Gaiotto, and S.~Gukov, ``{3-Manifolds and 3d Indices},''
  \href{http://dx.doi.org/10.4310/ATMP.2013.v17.n5.a3}{{\em Adv. Theor. Math.
  Phys.} {\bf 17} (2013) no.~5, 975--1076},
  \href{http://arxiv.org/abs/1112.5179}{{\tt arXiv:1112.5179 [hep-th]}}.

\bibitem{bailey1948identities}
W.~N. Bailey, ``Identities of the rogers-ramanujan type,''
  \href{http://dx.doi.org/https://doi.org/10.1112/plms/s2-50.1.1}{{\em
  Proceedings of the London Mathematical Society} {\bf s2-50} (1948) no.~1,
  1--10}.

\bibitem{slater1966generalized}
L.~Slater, {\em Generalized Hypergeometric Functions}.
\newblock Cambridge University Press, 1966.

\bibitem{bressoud1981}
D.~M. Bressoud, ``Some identities for terminating q-series,''
  \href{http://dx.doi.org/10.1017/S0305004100058114}{{\em Mathematical
  Proceedings of the Cambridge Philosophical Society} {\bf 89} (1981) no.~2,
  211–223}.

\bibitem{dbsearsbaileytransform}
G.~Andrews and D.~Bowman, ``{The Bailey transform and D.B. Sears},''
  \href{http://dx.doi.org/10.1080/16073606.1999.9632056}{{\em Quaestiones
  Mathematicae} {\bf 22} (1999)  19--26}.

\bibitem{Schilling1997}
A.~Schilling and S.~O. Warnaar, ``A higher-level bailey lemma,''
  \href{http://dx.doi.org/10.1142/s0217979297000253}{{\em International Journal
  of Modern Physics B} {\bf 11} (1997) no.~01n02, 189--195}.

\bibitem{geandrews}
G.~E. Andrews, ``{Multiple series Rogers-Ramanujan type identities.},'' {\em
  Pacific Journal of Mathematics} {\bf 114} (1984) no.~2, 267 -- 283.

\bibitem{milne1992al}
S.~C. Milne and G.~M. Lilly, ``{The $A_{\textit{l}}$ and $C_{\textit{l}}$
  Bailey transform and lemma},'' {\em Bull. Amer. Math. Soc.(NS)} {\bf 26}
  (1992)  258--263.

\bibitem{lilly1993thec}
G.~M. Lilly and S.~C. Milne, ``{The $C_{\textit{l}}$ Bailey transform and
  Bailey lemma},'' {\em Constructive Approximation} {\bf 9} (1993) no.~4,
  473--500.

\bibitem{milne1995consequences}
S.~C. Milne and G.~M. Lilly, ``Consequences of the a\textit{l} and c\textit{l}
  bailey transform and bailey lemma,'' {\em Discrete mathematics} {\bf 139}
  (1995) no.~1-3, 319--346.

\bibitem{Berkovich1996}
A.~Berkovich, B.~M. McCoy, and A.~Schilling, ``N = 2 supersymmetry and bailey
  pairs,'' {\em Physica A: Statistical Mechanics and its Applications} {\bf
  228} (1996) no.~1-4, 33--62.

\bibitem{Andrews1998scft}
G.~E. Andrews and A.~Berkovich, ``A trinomial analogue of bailey's lemma and n
  = 2 superconformal invariance,'' {\em Communications in Mathematical Physics}
  {\bf 192} (1998) no.~2, 245--260.

\bibitem{warnaarhistory}
S.~O. Warnaar, ``50 years of bailey's lemma,'' in {\em Algebraic Combinatorics
  and Applications}, pp.~333--347.
\newblock Springer Berlin Heidelberg, Berlin, Heidelberg, 2001.

\bibitem{Khachatryan:2012wy}
S.~Khachatryan and A.~Sedrakyan, ``{On the solutions of the Yang-Baxter
  equations with general inhomogeneous eight-vertex $R$-matrix: Relations with
  Zamolodchikov's tetrahedral algebra},''
  \href{http://dx.doi.org/10.1007/s10955-012-0666-8}{{\em J. Statist. Phys.}
  {\bf 150} (2013)  130}, \href{http://arxiv.org/abs/1208.4339}{{\tt
  arXiv:1208.4339 [math-ph]}}.

\bibitem{Karakhanyan:2013pv}
D.~Karakhanyan and S.~Khachatryan, ``{New solutions to the $sl_q(2)$ -invariant
  Yang-Baxter equations at roots of unity: Cyclic representations},''
  \href{http://dx.doi.org/10.1016/j.nuclphysb.2012.11.003}{{\em Nucl. Phys. B}
  {\bf 868} (2013)  328--349}, \href{http://arxiv.org/abs/1203.6528}{{\tt
  arXiv:1203.6528 [math-ph]}}.

\bibitem{Sarkissian:2020pwa}
G.~A. Sarkissian and V.~P. Spiridonov, ``{Rational hypergeometric
  identities},'' \href{http://arxiv.org/abs/2012.10265}{{\tt arXiv:2012.10265
  [math.CA]}}.

\end{thebibliography}\endgroup



\end{document}